\documentclass{siamart220329}

\usepackage{hyperref}

\usepackage{array}
\usepackage{float}
\usepackage{graphicx}
\usepackage{amstext}
\usepackage{amssymb}
\PassOptionsToPackage{normalem}{ulem}
\usepackage{ulem}

\usepackage{enumitem}
\setenumerate{label=\rm (\roman*)}

\usepackage{bbm}
\usepackage{centernot}

\usepackage{sherstov}

\numberwithin{theorem}{section}
\numberwithin{equation}{section}
\numberwithin{figure}{section}
\numberwithin{table}{section}

\allowdisplaybreaks[1]

\renewcommand{\mathcal}[1]{\mathscr{#1}}

\providecommand{\noopsort}[1]{}

\newcommand{\newsiamthmnonumber}[2]{
  \theoremstyle{nonumberplain}
  \theoremheaderfont{\normalfont\sc}
  \theorembodyfont{\normalfont\itshape}
  \theoremseparator{.}
  \theoremsymbol{}
  \newtheorem{#1}[theorem]{#2}
}

\newsiamthmnonumber{theoremnonumber}{Theorem}
\newsiamthm{fact}{Fact}
\newsiamremark{remark}{Remark}
\newsiamthm{claim}{Claim}

\ifpdf
  \DeclareGraphicsExtensions{.eps,.pdf,.png,.jpg}
\else
  \DeclareGraphicsExtensions{.eps}
\fi

\renewcommand{\mathcal}[1]{\mathscr{#1}}

\newcommand{\norm}[1]{|\!|\!| #1 | \! | \! |}
\newcommand{\NORM}[1]{\left|\!\left|\!\left| #1 \right| \! \right| \! \right|}

\newcommand{\pomo}{\{-1,1\}}

\newcommand{\pomon}{\pomo^n}

\renewcommand{\epsilon}{\varepsilon}

\newcommand{\TheLongTitle}{An Optimal Separation of Randomized and Quantum Query Complexity}
\newcommand{\TheRunningTitle}{Randomized versus Quantum Query Complexity}
\newcommand{\TheAuthors}{Alexander A. Sherstov, Andrey A. Storozhenko, and Pei Wu}

\title{{\TheLongTitle}\thanks{Submitted to the editors
January 15, 2022.  
A 14-page extended abstract of this paper appeared in \emph{Proceedings
of the 53rd Annual ACM Symposium on the Theory of Computing} (STOC),
2021, pages 1289\textendash 1302.
\funding{This work was supported in part by NSF grant CCF-1814947.}
}}

\author{Alexander A. Sherstov\thanks{University of California,
    Los Angeles, CA~90095
    (\email{sherstov@cs.ucla.edu}, 
     \email{storozhenko@cs.ucla.edu}).}
\and
Andrey A. Storozhenko\footnotemark[2]
\and
Pei Wu\thanks{Institute for Advanced Study, Princeton, NJ
08540 (\email{pwu@ias.edu}).
This work was done while the third author was a
Ph.D.~student at the University of California, Los
Angeles.}
}

\headers{\TheRunningTitle}{\TheAuthors}

\ifpdf
\hypersetup{
  pdftitle={\TheLongTitle},
  pdfauthor={\TheAuthors}
}
\fi

\begin{document}

\maketitle

\begin{abstract}
We prove that for every decision tree, the absolute values of the
Fourier coefficients of a given order $\ell\geq1$ sum to at most $c^{\ell}\sqrt{\binom{d}{\ell}(1+\log n)^{\ell-1}},$
where $n$ is the number of variables, $d$ is the tree depth, and
$c>0$ is an absolute constant. This bound is essentially tight and
settles a conjecture due to Tal~(arxiv 2019; FOCS 2020). The bounds
prior to our work degraded rapidly with $\ell,$ becoming trivial
already at $\ell=\sqrt{d}.$

As an application, we obtain, for every integer $k\geq1,$ a partial
Boolean function on $n$ bits that has bounded-error quantum query
complexity at most $k$ and randomized query complexity $\tilde{\Omega}(n^{1-\frac{1}{2k}}).$
This separation of bounded-error quantum versus randomized query complexity
is best possible, by the results of Aaronson and Ambainis~(STOC 2015)
and Bravyi, Gosset, Grier, and Schaeffer~(2021). Prior to our work, the best known separation
was polynomially weaker: $O(1)$ versus $\Omega(n^{2/3-\epsilon})$
for any $\epsilon>0$ (Tal, FOCS 2020).

As another application, we obtain an essentially optimal separation
of $O(\log n)$ versus $\Omega(n^{1-\epsilon})$ for bounded-error
quantum versus randomized communication complexity, for any $\epsilon>0.$
The best previous separation was polynomially weaker: $O(\log n)$
versus $\Omega(n^{2/3-\epsilon})$ (implicit in Tal,~FOCS~2020).
\end{abstract}



\section{Introduction}

Understanding the relative power of quantum and classical computing
is of basic importance in theoretical computer science. This question
has been studied most actively in the \emph{query model}, which is
tractable enough to allow unconditional lower bounds yet rich enough
to capture most of the known quantum algorithms. Illustrative examples
include the quantum algorithms of Deutsch and Jozsa~\cite{deutschjozsa1992rapid},
Bernstein and Vazirani~\cite{bernstein-vazirani93}, Grover~\cite{grover96search},
and Shor's period-finding~\cite{shor94factoring}. In the query model,
the task is to evaluate a fixed function $f$ on an unknown $n$-bit
input $x$. In the classical setting, query algorithms are commonly
referred to as \emph{decision trees. }A decision tree accesses the
input one bit at a time, choosing the bits to query in adaptive fashion.
The objective is to determine $f(x)$ by querying as few bits as possible.
The minimum number of queries needed to determine $f(x)$ in the worst
case is called the \emph{query complexity of $f$}. The quantum model
is a far-reaching generalization of the classical decision tree whereby
all bits can be queried in superposition with a single query. The
catch is that the outcomes of those queries are then also in superposition,
and it is not clear a priori whether quantum query algorithms are
more powerful than decision trees. We focus on the \emph{bounded-error}
regime, where the query algorithm (quantum or classical) is allowed
to err on any given input with probability $\epsilon$ for some constant
$\epsilon<1/2$.

The comparative power of randomized and quantum query algorithms has
been studied for more than two decades. In pioneering work, Deutsch
and Jozsa~\cite{deutschjozsa1992rapid} gave a quantum query algorithm
that solves, with a single query, a problem on $n$ bits that any
deterministic decision tree needs at least $n/2$ queries to solve.
Unfortunately, this separation does not apply to the more subtle,
bounded-error setting. This was addressed in follow-up work by Simon~\cite{simon97power},
who exhibited a problem with bounded-error quantum query complexity
$O(\log^{2}n)$ and randomized query complexity $\Omega(\sqrt{n})$.
These are striking examples of the computational advantages afforded
by the quantum model.

\subsection{Forrelation and rorrelation}

The above results leave us with a fundamental question: what is the
largest possible separation between bounded-error quantum and randomized
query complexity, for a problem with $n$-bit input? This question
was raised by Buhrman et al.~\cite{BFNR03quantum-property-testing}
and, a decade later, by Aaronson and Ambainis~\cite{AA15forrelation},
who presented it as being essential to understanding the phenomenon
of quantum speedups. Toward this goal, the authors of~\cite{AA15forrelation}
exhibited a problem that can be solved to bounded error with a single
quantum query but has randomized query complexity $\tilde{\Omega}(\sqrt{n}).$
They left open the challenge of obtaining a separation of $O(1)$
versus $\Omega(n^{\alpha})$ for some $\alpha>1/2.$ In more detail,
Aaronson and Ambainis~\cite{AA15forrelation} introduced and studied
the \emph{$k$-fold forrelation problem}. The input to the problem
is a $k$-tuple of vectors $x_{1},x_{2},\ldots,x_{k}\in\pomon,$ where
$n$ is a power of $2.$ Define 
\begin{equation}
\phi_{n,k}(x_{1},x_{2},\ldots,x_{k})=\frac{1}{n}\mathbf{1}^{\intercal}D_{x_{1}}HD_{x_{2}}HD_{x_{3}}H\cdots HD_{x_{k}}\mathbf{1},\label{eq:forrelation}
\end{equation}
where $\mathbf{1}$ is the all-ones vector, $H$ is the Hadamard transform
matrix of order $n$, and $D_{x_{i}}$ is the diagonal matrix with
the vector $x_{i}$ on the diagonal. Since each of the linear transformations
$H,D_{x_{1}},D_{x_{2}},\ldots,D_{x_{n}}$ preserves Euclidean length,
it follows that $|\phi_{n,k}(x_{1},x_{2},\ldots,x_{k})|\leq1.$ Given
$x_{1},x_{2},\ldots,x_{k},$ the forrelation problem is to distinguish
between the cases $|\phi_{n,k}(x_{1},x_{2},\ldots,x_{k})|\leq\alpha$
and $\phi_{n,k}(x_{1},x_{2},\ldots,x_{k})\geq\beta$, where the problem
parameters $0<\alpha<\beta<1$ are suitably chosen constants. Equation~(\ref{eq:forrelation})
directly gives a quantum algorithm that solves the forrelation problem
with bounded error and query cost $k,$ where the $k$ queries correspond
to the $k$ diagonal matrices. The cost can be further reduced to
$\lceil k/2\rceil$ by viewing~(\ref{eq:forrelation}) as the \emph{inner
product} of two vectors obtained by $\lceil k/2\rceil$ and $\lfloor k/2\rfloor$
applications, respectively, of diagonal matrices~\cite{AA15forrelation}.
Aaronson and Ambainis complemented this with an $\tilde{\Omega}(\sqrt{n})$
lower bound on the randomized query complexity of the forrelation
problem for $k=2$, hence the $1$ versus $\tilde{\Omega}(\sqrt{n})$
separation mentioned above.

Tal~\cite{Tal19Query} built on \cite{AA15forrelation} to give an
improved separation of $O(1)$ versus $\Omega(n^{2/3-\epsilon})$
for bounded-error quantum and randomized query complexities, for any
constant $\epsilon>0.$ For this, Tal replaced~(\ref{eq:forrelation})
with the more general quantity
\begin{equation}
\phi_{n,k,U}(x_{1},x_{2},\ldots,x_{k})=\frac{1}{n}\mathbf{1}^{\intercal}D_{x_{1}}UD_{x_{2}}UD_{x_{3}}U\cdots UD_{x_{k}}\mathbf{1},\label{eq:rorrelation}
\end{equation}
where $U$ is an arbitrary but fixed orthogonal matrix. On input $x_{1},x_{2},\ldots,x_{k}\in\pomon,$
the author of~\cite{Tal19Query} considered the problem of distinguishing
between the cases $|\phi_{n,k,U}(x_{1},x_{2},\ldots,x_{k})|\leq2^{-k-1}$
and $\phi_{n,k,U}(x_{1},x_{2},\ldots,x_{k})\geq2^{-k}.$ This problem
is referred to in~\cite{Tal19Query} as the \emph{$k$-fold rorrelation
problem with respect to $U.$} The quantum algorithm of Aaronson and
Ambainis, adapted to the arbitrary choice of $U,$ solves this new
problem with $\lceil k/2\rceil$ queries and advantage $\Omega(2^{-k})$
over random guessing, which counts as a bounded-error algorithm for
any constant $k.$ On the other hand, Tal~\cite{Tal19Query} proved
that the randomized query complexity of the $k$-fold rorrelation
problem for uniformly random $U$ is $\Omega(n^{2(k-1)/(3k-1)}/k\log n)$
with high probability. While this is weaker than Aaronson and Ambainis's
bound for $k=2,$ setting $k$ to a large constant gives a separation
of $O(1)$ versus $\Omega(n^{2/3-\epsilon})$ for bounded-error quantum
and randomized query complexity for any constant $\epsilon>0$.

\subsection{Our results: separations for partial functions}
Prior to our paper, Tal's separation of $O(1)$ versus $\Omega(n^{2/3-\epsilon})$
was the strongest known, and Aaronson and Ambainis's challenge of
obtaining an $O(1)$ versus $\Omega(n^{1-\epsilon})$ separation remained
open. The main contribution of our work is to resolve this question.
In what follows, we let $f_{n,k,U}$ denote the $k$-fold rorrelation
problem with respect to $U.$ We prove:
\begin{theorem}
\label{thm:MAIN-randomized-lower-bound} Let $n$ be a positive
integer. Let $U\in\Re^{n\times n}$
be a uniformly random orthogonal matrix. Then with probability $1-o(1)$ 
over the choice of $U,$ one has
\begin{align}
R_{\frac{1}{2}-\gamma}(f_{n,k,U}) & =\Omega\left(\frac{\gamma^{2}}{k}\cdot\frac{n^{1-\frac{1}{k}}}{(\log n)^{2-\frac{1}{k}}}\right)\label{eq:R-lower-bound-general-1}
\end{align}
for all integers $k\leq\frac{1}{3}\log n-1$ and
all $0\leq\gamma\leq1/2.$
\end{theorem}

\noindent For $k=2,$ this lower bound is the same (up
to a $\sqrt{\log n}$ factor) as Aaronson and
Ambainis's lower bound for the forrelation problem (which is $f_{n,2,H}$
in our notation). For $k=3$ already, Theorem~\ref{thm:MAIN-randomized-lower-bound}
is a polynomial improvement on all previous work, including Tal's
recent result~\cite{Tal19Query}. Up to logarithmic factors, Theorem~\ref{thm:MAIN-randomized-lower-bound}
is tight for every $k$ due to the matching upper bound $O_{k}(n^{1-1/k})$
of Bravyi, Gosset, Grier, and Schaeffer~\cite[Theorem~6]{bravyi2021quantum-to-classical}.

For any constant $k,$ the rorrelation problem $f_{n,k,U}$ has a
bounded-error quantum query algorithm with cost $\lceil k/2\rceil$
(see Section~\ref{subsec:rorrelation} for details). As a result,
by taking $k=2t$ for an integer $t,$ we obtain the following separation of bounded-error quantum and randomized
query complexities.
\begin{corollary}
\label{cor:bounded-error-t-separation}Let $t\geq1$ be a fixed integer.
Then there is a partial Boolean function $f$ on $\pomon$ with
\begin{align*}
 & Q_{\frac{1}{2}-\Omega(1)}(f)\leq t,\\
 & R_{1/3}(f)=\Omega\left(\frac{n^{1-\frac{1}{2t}}}{(\log n)^{2-\frac{1}{2t}}}\right).
\end{align*}
\end{corollary}

\noindent The separation in Corollary~\ref{cor:bounded-error-t-separation}
is best possible. Indeed, Bravyi~et al.~\cite[Theorems~3 and~6]{bravyi2021quantum-to-classical}
show that for every constant $t,$ every quantum algorithm with $t$
queries can be converted into a randomized decision tree of cost $O(\frac{1}{\epsilon}\cdot n^{1-\frac{1}{2t}})$
whose acceptance probabilities for all inputs are within an additive $\epsilon$
of the quantum algorithm's corresponding acceptance probabilities.
Taking $t$ large in Corollary~\ref{cor:bounded-error-t-separation}
gives the following clean result:
\begin{corollary}
\label{cor:MAIN-bounded-separation}Let $\epsilon>0$ be given. Then
there is a partial Boolean function $f$ on $\pomon$ with
\begin{align*}
Q_{1/3}(f) & =O(1),\\
R_{1/3}(f) & =\Omega(n^{1-\epsilon}).
\end{align*}
\end{corollary}

\noindent Again, this separation of bounded-error quantum and randomized
query complexities is best possible for all $f$ due to the aforementioned
result of Bravyi et al.~that every quantum algorithm with $t$ queries
can be simulated by a randomized query algorithm of cost $O_{t}(n^{1-\frac{1}{2t}})$.
In particular, Corollary~\ref{cor:MAIN-bounded-separation} shows
that the rorrelation problem separates quantum and randomized query
complexity optimally, of all problems $f$. The following incomparable
corollary can be obtained by taking $k=k(n)$ in Theorem~\ref{thm:MAIN-randomized-lower-bound}
to be an arbitrarily slow-growing function, e.g., $k=\log\log\log n$:

\begin{corollary}
\label{cor:MAIN-near-linear-separation}Let $\alpha\colon\NN\to\NN$
be any monotone function with $\alpha(n)\to\infty$ as $n\to\infty.$
Then there is a partial Boolean function $f$ on $\pomon$ with
\begin{align*}
Q_{1/3}(f) & \leq\alpha(n),\\
R_{1/3}(f) & \geq n^{1-o(1)}.
\end{align*}
\end{corollary}

\noindent As before, this quantum-classical separation is best possible
since~\cite{bravyi2021quantum-to-classical} rules out the possibility
of an $O(1)$ versus $n^{1-o(1)}$ gap.

A satisfying probability-theoretic interpretation of our results is
that the phenomenon of quantum-classical gaps is a common one. More
precisely, our results show that the set of orthogonal matrices $U$
for which $f_{n,k,U}$ does \emph{not} exhibit a best-possible quantum-classical
separation has Haar measure $0.$ Prior to our work, this was unknown
for any integer $k>2.$

\subsection{Our results: separation for total functions}

Our results so far pertain to \emph{partial} Boolean functions, whose
domain of definition is a proper subset of the Boolean hypercube.
For total Boolean functions, such large quantum-classical gaps
are not possible. In a seminal paper, Beals et al.~\cite{beals-et-al01quantum-by-polynomials}
prove that the bounded-error quantum query complexity of a total function
$f$ is always polynomially related to the randomized query complexity
of $f$. A natural question to ask is how large this polynomial gap
can be. Grover's search~\cite{grover96search} shows that the $n$-bit
OR function has bounded-error quantum query complexity $\Theta(\sqrt{n})$
and randomized complexity $\Theta(n).$ For a long time, this quadratic
separation was believed to be the largest possible. In a surprising
result, Aaronson et al.~\cite{ABK16cheat} proved the existence of
a total function $f$ with $R_{1/3}(f)=\tilde{\Omega}(Q_{1/3}(f)^{2.5}).$
This was improved by Tal~\cite{Tal19Query} to $R_{1/3}(f)\geq Q_{1/3}(f)^{8/3-o(1)}.$
We give a polynomially stronger separation:
\begin{theorem}
\label{thm:MAIN-total}There is a function $f\colon\pomon\to\zoo$
with
\[
R_{1/3}(f)\geq Q_{1/3}(f)^{3-o(1)}.
\]
\end{theorem}

Theorem~\ref{thm:MAIN-total} follows immediately by combining
our Corollary~\ref{cor:MAIN-near-linear-separation} with the ``cheatsheet''
framework of Aaronson et al.~\cite{ABK16cheat}. Specifically, they
prove that any partial function $f$ on $n$ bits that exhibits an
$n^{o(1)}$ versus $n^{1-o(1)}$ separation for bounded-error quantum
versus randomized query complexity, can be automatically converted
into a total function with $R_{1/3}(f)\geq Q_{1/3}(f)^{3-o(1)}.$
A recent paper of Aaronson et al.~\cite{ABKT20} conjectures that
$R_{1/3}(f)=O(Q_{1/3}(f)^{3})$ for every total function $f,$ which
would mean that our separation in Theorem~\ref{thm:MAIN-total} is
essentially optimal. The best current upper bound is $R_{1/3}(f)=O(Q_{1/3}(f)^{4})$
due to~\cite{ABKT20}, derived there from the breakthrough result
of Huang~\cite{huang2019induced} on the sensitivity conjecture.

\subsection{Our results: separations for communication complexity}

Via standard reductions, our quantum-classical query separations
imply analogous separations for communication complexity. In more
detail, let $f$ be a (possibly partial) Boolean function on $\pomon.$
For any communication problem $g\colon\pomo^{m}\times\pomo^{m}\to\pomo,$
we let $f\circ g$ stand for the (possibly partial) communication problem
on $(\pomo^{m})^{n}\times(\pomo^{m})^{n}$ given by $(f\circ g)(x,y)=f(g(x_{1},y_{1}),g(x_{2},y_{2}),\ldots,g(x_{n},y_{n})).$
Buhrman, Cleve, and Wigderson~\cite{bcw98quantum} proved that any
quantum query algorithm for $f$ gives a quantum communication protocol
for $f\circ g$ with the same error and approximately the same cost.
Quantitatively, 
\begin{equation}
Q_{\epsilon}^{\text{cc}}(f\circ g)\leq Q_{\epsilon}(f)\cdot O(m+\log n),\label{eq:Qcc-Q}
\end{equation}
where $Q_{\epsilon}^{\text{cc}}$ denotes $\epsilon$-error quantum
communication complexity. Reversing this inequality has seen a great
deal of work, mainly in the classical setting. A well-studied function
$g$ in this line of research is the inner product function $\IP_{m}\colon\pomo^{m}\times\pomo^{m}\to\pomo,$
given by $\IP_{m}(u,v)=\bigoplus_{i=1}^{m}(u_{i}\wedge v_{i}).$ In
particular, Chattopadhyay, Filmus, Koroth, Meir, and Pitassi~\cite[Theorem~1]{cfkmp19lifting-randomized-with-IP-gadget}
prove that
\begin{equation}
R_{1/3}^{\text{cc}}(f\circ\IP_{c\log n})=\Omega(R_{1/3}(f)\log n)\label{eq:Rcc-R}
\end{equation}
for every (possibly partial) function $f$ on $\pomo^{n},$ where
$R_{\epsilon}^{\text{cc}}$ denotes $\epsilon$-error randomized communication
complexity and $c>1$ is an absolute constant. In light of this connection
between query complexity and communication complexity, our main results
have the following consequences.
\begin{theorem}
\label{thm:MAIN-bounded-separation-comm}Let $\epsilon>0$ be given.
Then there is a partial Boolean function $F$ on $\pomo^{N}\times\pomo^{N}$
with
\begin{align*}
Q_{1/3}^{\text{\emph{cc}}}(F) & =O(\log N),\\
R_{1/3}^{\text{\emph{cc}}}(F) & =\Omega(N^{1-\epsilon}).
\end{align*}
\end{theorem}

\begin{proof}
Take $f$ as in Corollary~\ref{cor:MAIN-bounded-separation} and
define $N=cn\log n$ and $F=f\circ\IP_{c\log n}.$ Then the communication
bounds follow from~(\ref{eq:Qcc-Q}) and~(\ref{eq:Rcc-R}), respectively.
\end{proof}
\noindent Theorem~\ref{thm:MAIN-bounded-separation-comm} is essentially
optimal and a polynomial improvement on previous work. The best previous
quantum-classical separation for communication complexity was $O(\log N)$
versus $\Omega(N^{2/3-\epsilon})$, implicit in Tal~\cite{Tal19Query}
and preceded in turn by other exponential separations~\cite{raz99quantum-classical,regev-klartag11quantum-vs-classical,gavinsky20quantum-vs-classical-comm}.
Similarly, our Corollary~\ref{cor:MAIN-near-linear-separation} translates
in a black-box manner to communication complexity:
\begin{theorem}
\label{thm:MAIN-near-linear-separation-communication}Let $\alpha\colon\NN\to\NN$
be any monotone function with $\alpha=\omega(1).$ Then there is a
partial Boolean function $F$ on $\pomo^{N}\times\pomo^{N}$ with
\begin{align*}
Q_{1/3}^{\text{\emph{cc}}}(F) & \leq\alpha(N)\log N,\\
R_{1/3}^{\text{\emph{cc}}}(F) & \geq N^{1-o(1)}.
\end{align*}
\end{theorem}

\begin{proof}
Take $f$ as in Corollary~\ref{cor:MAIN-near-linear-separation}
and define $N=cn\log n$ and $F=f\circ\IP_{c\log n}.$ Then the communication
bounds follow from~(\ref{eq:Qcc-Q}) and~(\ref{eq:Rcc-R}), respectively.
\end{proof}
Finally, we obtain the following result for \emph{total} functions.
\begin{theorem}
\label{thm:MAIN-total-comm}There is a function $F\colon\pomo^{N}\times\pomo^{N}\to\zoo$
with
\[
R_{1/3}^{\text{\emph{cc}}}(F)\geq Q_{1/3}^{\text{\emph{cc}}}(F)^{3-o(1)}.
\]
\end{theorem}

\begin{proof}
The cheatsheet framework~\cite{ABK16cheat} ensures that the quantum
and classical query complexities of $f$ in Theorem~\ref{thm:MAIN-total}
are polynomial in the number of variables $n.$ With this in mind,
we proceed as before, setting $N=cn\log n$ and $F=f\circ\IP_{c\log n}$
and applying~(\ref{eq:Qcc-Q}) and~(\ref{eq:Rcc-R}).
\end{proof}
\noindent Again, Theorem~\ref{thm:MAIN-total-comm} is a polynomial
improvement on previous work, the best previous result being a power
of $8/3$ separation implicit in~\cite{Tal19Query}.

\subsection{Our results: Fourier weight of decision trees}

It is straightforward to verify that a uniformly random input $x\in(\pomon)^{k}$
is with high probability a \emph{negative} instance of the rorrelation
problem $f_{n,k,U}$. With this in mind, Tal~\cite{Tal19Query} proves
his lower bound for rorrelation by constructing a probability distribution
$\Dcal_{n,k,U}$ that generates \emph{positive} instances of $f_{n,k,U}$
with nontrivial probability yet is indistinguishable from the uniform
distribution by a decision tree $T$ of cost $n^{2/3-O(1/k)}$. His
notion of indistinguishability is based on the Fourier spectrum. Specifically,
Tal~\cite{Tal19Query} shows that: (i) the \emph{sum} of the absolute
values of the Fourier coefficients of $T$ of a given order $\ell$
does not grow too fast with $\ell$; and (ii) the \emph{maximum} Fourier
coefficient of $\Dcal_{n,k,U}$ of order $\ell$ decays exponentially
fast with $\ell$. In Tal's paper, the bound for (ii) is essentially
optimal, whereas the bound for (i) is far from tight. The sum of the
absolute values of the order-$\ell$ Fourier coefficients of a decision
tree $T$, which we refer to as the \emph{$\ell$-Fourier weight of
$T$}, is shown in~\cite{Tal19Query} to be at most
\begin{equation}
c^{\ell}\sqrt{d^{\ell}(1+\log kn)^{\ell-1},}\label{eq:tals-bound-fourier-weight}
\end{equation}
where $d$ is the depth of the tree and $c\geq1$ is an absolute constant.
This bound is strong for any constant $\ell$ but degrades rapidly
as $\ell$ grows. In particular, for $\ell=\sqrt{d}$ already, (\ref{eq:tals-bound-fourier-weight})
is weaker than the trivial bound $\binom{d}{\ell}.$ This is a major
obstacle since the indistinguishability proof requires strong bounds
for every $\ell$. This obstacle is the reason why Tal's analysis
gives the randomized query lower bound $n^{2/3-O(1/k)}$ as opposed
to the optimal $\tilde{\Omega}(n^{1-1/k}).$ Tal conjectured that
the $\ell$-Fourier weight of a depth-$d$ decision tree is in fact
bounded by $c^{\ell}\sqrt{\binom{d}{\ell}(1+\log kn)^{\ell-1}},$
which is a factor of $\sqrt{\ell!}$ improvement on~(\ref{eq:tals-bound-fourier-weight})
and essentially optimal. We prove his conjecture:
\begin{theorem}
\label{thm:MAIN-Fourier-bound}Let $T\colon\pomon\to\{0,1\}$ be a
function computable by a decision tree of depth $d.$ Then
\begin{align*}
\sum_{\substack{S\subseteq\oneton:\\
|S|=\ell
}
}|\hat{T}(S)| & \leq c^{\ell}\sqrt{\binom{d}{\ell}(1+\log n)^{\ell-1}}, &  & \ell=1,2,\ldots,n,
\end{align*}
where $c\geq1$ is an absolute constant.
\end{theorem}

\noindent It is well known and easy to show that Theorem~\ref{thm:MAIN-Fourier-bound}
is essentially tight, even for \emph{nonadaptive }decision trees~\cite[Theorem~5.19]{odonnell14boolean-fuction-analysis}.
The actual statement that we prove is more precise and takes into
account the density parameter $\Prob[T(x)\ne0]$; see Theorem~\ref{thm:main-fourier-DT-p}
for details. With Theorem~\ref{thm:MAIN-Fourier-bound} in hand,
all our main results (Theorem~\ref{thm:MAIN-randomized-lower-bound}
and its corollaries) follow immediately by combining the new bound
on the Fourier weight of decision trees with Tal's near-optimal bounds
on the individual Fourier coefficients of $\Dcal_{n,k,U}.$

Theorem~\ref{thm:MAIN-Fourier-bound} is of interest in its own right,
independent of its use in this paper to obtain optimal quantum-classical
separations. The study of the Fourier spectrum has a variety of applications
in theoretical computer science, including circuit complexity, learning
theory, pseudorandom generators, and quantum computing. Even prior
to Tal's work, the $\ell$-Fourier weight of decision trees was studied
for $\ell=1$ by O'Donnell and Servedio~\cite{OS07}, who proved
the tight $O(\sqrt{d})$ bound and used it to give a polynomial-time
learning algorithm for monotone decision trees. Fourier weight has
been studied for various other classes of Boolean functions, including
bounded-depth circuits, branching programs, low-degree polynomials
over finite fields, and functions with bounded sensitivity; 
see~\cite{GSW16sensitivity,SVW17,Tal17ac0,CHRT18ROBP,CHHL18walk,BIJLSV21f2poly,LPV22branchingprogram}
and the references therein.

\subsection{Limitations of previous analyses}

In this part, we overview Tal's bound on the $\ell$-Fourier weight
of decision trees. To build intuition, it is helpful to first examine
the case $\ell=1$, due to O'Donnell and Servedio~\cite{OS07} and
Tal~\cite{Tal19Query}. For simplicity, consider a perfect tree $T$
of depth $d$ with leaves labeled $0$ and $1,$ where the $i$-th
variable queried in each path is $x_{i}.$ Throughout this discussion,
we identify a decision tree with the function that it computes, and
use the same variable for both. By negating the variables if necessary,
we may assume that $\hat{T}(i)\ge0$. In particular,
\[
\sum_{i=1}^{n}|\hat{T}(i)|=\Exp_{x}\left[T(x)\sum_{i=1}^{d}x_{i}\right].
\]
This gives a new perspective on $\sum|\hat{T}(i)|$ in terms of the
random experiment whereby one picks a random root-to-leaf path, sums
all the variables in that path, and multiplies the result by the label
of the leaf. The expected value of this experiment equals $\sum|\hat{T}(i)|.$
It is clear that this value is maximized when the leaves labeled $1$
correspond to paths with large sums. With this observation~\cite{Tal19Query},
one can prove that
\begin{equation}
\sum_{i=1}^{n}|\hat{T}(i)|=O\left(p\sqrt{d\ln\frac{e}{p}}\right),\label{eq:intro-Tal-level-1-ineq}
\end{equation}
where $p=\Prob[T(x)\ne0]$ is the fraction of nonzero leaves, which
we refer to as the \emph{density of $T$}. By linearity, the same
argument applies even to adaptive trees.

Tal's analysis for $\ell\geq2$ is a natural inductive generalization
of the above argument. Let $T$ be an arbitrary tree in variables
$x_{1},x_{2},\ldots,x_{n}.$ Let $V_{i}$ denote the set of internal
nodes in $T$ labeled by the variable $x_{i}$. The key notion is
that of the \emph{contraction of $T$ with respect to $x_{i},$} which
is a tree denoted by $T_{i}$ with real-valued labels at the leaves.
This tree $T_{i}$ is formed by the following two-step process: (i)~for
each path that does not query $x_{i}$, set the leaf label to $0;$
and (ii) for each $v\in V_{i}$, replace the subtree $T_{v}$ rooted
at $v$ by a single leaf labeled by the Fourier coefficient $\hat{T}_{v}(i).$
The $n$ contractions of $T$ give rise to the decomposition
\begin{equation}
\sum_{|S|=\ell}|\hat{T}(S)|\le\sum_{i=1}^{n}\sum_{|S|=\ell-1}|\hat{T}_{i}(S)|,\label{eq:intro-Tal-bridge-inductive-hypothesis}
\end{equation}
which is the foundation of Tal's inductive argument. The real-valued
labels of the $T_{i}$ present no difficulty since one can replace
each such label by its binary expansion and thus write $T_{i}$ as
a linear combination of trees with binary labels. The key parameter
in Tal's inductive proof is density, and it needs to be maintained
carefully for each of the trees involved. Since the contractions of
$T$ can overlap in complicated ways, it becomes increasingly difficult
to accurately keep track of the densities. This translates into progressively
larger losses at each step of the inductive argument. Cumulatively,
the argument incurs an extraneous factor of $\sqrt{\ell!}$ in the
final bound. Despite considerable efforts, we were not able to find
a way forward within this framework.

\subsection{Our approach}

To obtain the near-optimal bound in Theorem~\ref{thm:MAIN-Fourier-bound},
we adopt a completely different approach. At a high level, we partition
$\sum_{|S|=\ell}|\hat{T}(S)|$ into well-structured parts. We discuss
the partitioning strategy first, and then our analysis of each part
in the partition.

\subsubsection*{The partition}

Let $T$ be a perfect tree of depth $d.$ We think of the vertices
at any given depth as forming a \emph{layer}, and we number the layers
of $T$ consecutively $1$ through $d.$ Consider a grouping of the
layers into $\ell$ disjoint blocks $I_{1},I_{2},\ldots,I_{\ell}\subseteq\{1,2,\ldots,d\}$,
where each block consists of consecutive layers from $T$, and the
union $I_{1}\cup I_{2}\cup\cdots\cup I_{\ell}$ may be a proper subset
of $\{1,2,\ldots,d\}$. The numbering of these blocks
corresponds to the ordering of the elements,  i.e., the
elements of $I_1$ are all less than the elements of
$I_2,$ which are in turn less than the elements of
$I_3,$ and so on.
As a canonical example, we could partition
the layers into $\ell$ blocks of roughly equal size. Viewed as a
function, $T$ is the sum of the characteristic functions of the root-to-leaf
paths, each such path weighted by the corresponding leaf. If one alters
this sum by keeping, for each path, only those Fourier coefficients
that have exactly one variable in each block, the result is a real-valued
function which we denote by $T|_{I_{1}*I_{2}*\cdots*I_{\ell}}.$ Here
we define $I_{1}\ast I_{2}\ast\cdots\ast I_{\ell}=\{S\in\binom{[d]}{\ell}:|S\cap I_{i}|=1\text{ for each }i\}$,
and we refer to any such family of sets in $\binom{[d]}{\ell}$ as
an \emph{elementary family.} Our challenge is to find an efficient
partition of $\binom{[d]}{\ell}$ into elementary families $\Ecal_{1},\Ecal_{2},\ldots,\Ecal_{N}$.
Then
\begin{equation}
T|_{\binom{[d]}{\ell}}=\sum_{i=1}^{N}T|_{\Ecal_{i}},\label{eq:partition-fourier-spectrum}
\end{equation}
and we can bound the Fourier weight of the degree-$\ell$ homogeneous
part of $T$ by bounding that of $T|_{\Ecal_{i}}$ for each $i$.
For the proof of Theorem~\ref{thm:MAIN-Fourier-bound}, we need a
partition that achieves
\begin{equation}
\sum_{i=1}^{N}\sqrt{|\Ecal_{i}|}\le C^{\ell}\sqrt{\binom{d}{\ell}}\label{eq:intro-efficient-partition}
\end{equation}
for an absolute constant $C\geq1.$ Such a partition would be essentially
extremal due to the trivial lower bound $\sum\sqrt{|\Ecal_{i}|}\geq\sqrt{\sum|\Ecal_{i}|}=\binom{d}{\ell}^{1/2}$
for every partition of $\binom{[d]}{\ell}$. Unfortunately, with elementary
families defined as above, such a partition does not exist! For the
sake of simplicity, we ignore this complication altogether in the
remainder of this discussion. In the actual proof, we resolve this
issue by allowing elementary families to contain up to two variables
per block. This makes the rest of the proof more delicate, but still
suffices for the purposes of proving Theorem~\ref{thm:MAIN-Fourier-bound}.
We give a first-principles combinatorial construction of a partition
with~(\ref{eq:intro-efficient-partition}) in Section~\ref{sec:Elementary-set-families}.

\subsubsection*{Analysis of individual parts}

For any elementary family $\Ecal,$
we prove that $T|_{\Ecal}$ has $\ell$-Fourier weight
\begin{equation}
\sqrt{|\Ecal|\cdot O(\log n)^{\ell-1}}.\label{eq:intro-level-k-ineq}
\end{equation}
Along with~(\ref{eq:partition-fourier-spectrum}) and~(\ref{eq:intro-efficient-partition}),
this implies Theorem~\ref{thm:MAIN-Fourier-bound}. 
Indeed, applying the bound just claimed to each summand in
(\ref{eq:partition-fourier-spectrum}) shows that the decision tree has $\ell$-Fourier weight
$\sum_{i=1}^N \sqrt{|\mathcal{E}_i|\cdot O(\log n)^{\ell-1}},$ which
by~(\ref{eq:intro-efficient-partition}) is at most 
$C^\ell \sqrt{\binom{d}{\ell}\cdot O(\log n)^{\ell-1}}.$

In this overview, we focus on proving (\ref{eq:intro-level-k-ineq}) 
for the special case
$\Ecal=I_{1}\ast I_{2}\ast\cdots\ast I_{\ell}$
with
\[
|I_{1}|=|I_{2}|=\cdots=|I_{\ell}|=\frac{d}{\ell}.
\]
Our bound (\ref{eq:intro-level-k-ineq}) uses a generalization of
decision trees where the leaves can be labeled by polynomials. With
this generalization, we can further define tree addition, as well
as tree multiplication by polynomials. This provides a powerful framework
for decomposing trees and expressing them as conical combinations
of simpler trees. To see how this generalization comes into play,
consider the subtree $T_{v}$ rooted at some node $v$ in the \emph{first}
layer of $I_{\ell}$. By the structure of $T|_{\Ecal},$ the only
relevant aspect of $T_{v}$ is its degree-1 homogeneous part. Therefore,
$T_{v}$ can be \emph{replaced} with its degree-$1$ homogeneous part.
Now, let $T'$ be the decision tree obtained by contracting every
node $v$ in the first layer of $I_{\ell}$ into a leaf labeled by
the polynomial $\sum_{i=1}^{n}\hat{T}_{v}(i)x_{i}.$ We show that
analyzing the Fourier weight of $T|_{I_{1}\ast I_{2}\ast\cdots\ast I_{\ell}}$
is equivalent to analyzing that of $T'$ with respect to the smaller
elementary family $I_{1}\ast I_{2}\ast\cdots\ast I_{\ell-1}.$ The
latter is a delicate task, and our solution involves three stages.
\begin{enumerate} \item \label{three-part-program}
In the first stage, we group leaves $v$ in $T'$ according to the
density $\alpha_{v}$ of the original subtree $T_{v}$. Applying Tal's argument~\cite{Tal19Query}, we have
\[
\sum_{i=1}^{n}|\hat{T}_{v}(i)|\le c'\alpha_{v}\sqrt{\frac{d}{\ell}\ln\frac{e}{\alpha_{v}}}
\]
for some constant $c'\geq1.$ We decompose $T'=\sum_{j=0}^{\infty}T'_{j}$,
where $T_{j}'$ keeps a leaf $v$ if $\alpha_{v}\in(3^{-j-1},3^{-j}]$
and replaces it with 0 otherwise.
\item In the second stage, we further decompose $T'_{j}$ as follows. Let
$\beta_{j}$ be the fraction of nonzero leaves in $T_{j}',$ and let
$m$ be the maximum Fourier weight of a nonzero leaf $v$ of $T'_{j}.$
We then express $T_{j}'$ as the conical combination $T'_{j}=\sum_{r=1}^{\infty}c_{r}T'_{j,r}$
such that: $\sum c_{r}=m$; each nonzero leaf of $T_{j,r}'$ is labeled
with some variable or its negation; and the fraction of nonzero leaves
in each $T_{j,r}'$ is $\beta_{j}$.
\item In the final stage, we decompose $T'_{j,r}$ into $n$ different trees
according to the $n$ variables: $T'_{j,r}=\sum_{i=1}^{n}T'_{j,r,i}\cdot x_{i}.$
The tree $T'_{j,r,i}$ keeps only those leaves $v$ that are labeled
by $\pm x_{i}$, and the new label is exactly the sign of the variable
$x_{i}$. Now $T'_{j,r,i}:\{-1,1\}^{n}\to\{-1,0,1\}$ has density
$\beta_{j}/n$ on average, and $T'_{j,r,i}|_{I_{1}\ast I_{2}\ast\cdots\ast I_{\ell-1}}$
can be analyzed using the inductive hypothesis.
\end{enumerate}
Of the three stages, the first stage is the least natural but crucial.
To see this, let $\ell=2$ and consider the following extreme case:
for all nonzero leaves $v$ in $T'$, the densities $\alpha_{v}$
are equal, $\alpha_{v}=\alpha$. Let $p$ denote the density of $T.$
Observe that $p$ is the product of $\alpha$ and the density of $T',$ which means that $T'$ has density $p/\alpha.$
There is some $j$ such that $T'=T_{j}'$, and that specific $T'_{j}$ has
density $p/\alpha$. Consequently, $T'_{j,r,i}$ has density $p/(n\alpha)$
on average. The $1$-Fourier weight of $T'_{j,r,i}$ for average $i$
can be bounded by 
\[
c'\cdot\frac{p}{n\alpha}\sqrt{\frac{d}{2}\ln\frac{en\alpha}{p}}.
\]
The Fourier weight of $T'|_{\{1,2,\ldots,d/2\}\ast\{d/2+1,d/2+2,\ldots,d\}}$
can then be bounded by 
\begin{align}
c'\cdot\alpha\sqrt{\frac{d}{2}\ln\frac{e}{\alpha}} & \cdot\sum_{i=1}^{n}c'\cdot\frac{p}{n\alpha}\sqrt{\frac{d}{2}\ln\frac{en\alpha}{p}}\nonumber \\
 & =(c')^{2}\cdot p\sqrt{\left(\frac{d}{2}\right)^{2}\ln\frac{e}{\alpha}\cdot\ln\frac{en\alpha}{p}}.\label{eq:intro-our-bound-level-2}
\end{align}
The corresponding bound for $\ell=2$ that Tal obtains is 
\[
O\left(p\sqrt{d^{2}\ln\frac{e}{p}\cdot\ln\frac{en}{p}}\right).
\]
Comparing it with our bound~(\ref{eq:intro-our-bound-level-2}) shows
that for $\alpha\gg p,$ our factor $\ln\frac{e}{\alpha}$ is substantially
smaller than Tal's corresponding factor $\ln\frac{e}{p}$; while for
$\alpha$ close to $p$, our factor $\ln\frac{en\alpha}{p}$ is substantially
smaller than Tal's $\ln\frac{en}{p}$. For
$\ell=\omega(1)$, the savings become significant. This is the intuitive reason
why the first stage allows us to avoid the $\sqrt{\ell!}$ loss. Its
surprising power comes from the framework of elementary families set
up at the beginning of the proof.

Our complete analysis of the Fourier weight of decision
trees is presented in Section~\ref{sec:fourier-main}.
Sections~\ref{subsec:Slicing-the-tree}
and~\ref{subsec:analytic-preliminaries} supply
Fourier-theoretic and analytic preliminaries.
In Section~\ref{subsec:Contiguous-intervals}, we study
the Fourier weight of decision trees with respect to
elementary families of special form, as in the proof
overview above. These results are generalized to
arbitrary elementary families in
Section~\ref{subsec:Generalization-to-elementary}.  Our
main result on the Fourier weight of decision trees is
then established in Section~\ref{subsec:Main-result}.
The concluding Section~\ref{sec:quantum} leverages
these contributions to establish our main results on
quantum versus classical query complexity.

\subsection{Independent work by Bansal and Sinha}

Independently and concurrently with our work, Bansal and Sinha~\cite{bansal-sinha20forrelation}
also obtained an optimal, $\lceil k/2\rceil$ versus $\tilde{\Omega}(n^{1-1/k})$
separation of quantum and randomized query complexity. Their result
uses completely different techniques and is incomparable with ours.
In more detail, Bansal and Sinha~\cite{bansal-sinha20forrelation}
construct a function $f$ with randomized query complexity
\begin{align}
R_{\frac{1}{2}-\gamma}(f) & =\Omega\left(\frac{\gamma^{2}}{k^{29}}\cdot\left(\frac{n}{\log(k+n)}\right)^{1-\frac{1}{k}}\right), &  & \forall\gamma\in[0,1/2].\label{eq:bansal-sinha}
\end{align}
This is essentially the same as our lower bound on randomized query
complexity (Theorem~\ref{thm:MAIN-randomized-lower-bound}):
\begin{align*}
R_{\frac{1}{2}-\gamma}(f_{n,k,U}) & =\Omega\left(\frac{\gamma^{2}}{k}\cdot\frac{n^{1-\frac{1}{k}}}{(\log n)^{2-\frac{1}{k}}}\right), & \forall\gamma\in[0,1/2].
\end{align*}
In both cases, the function in question has a quantum query algorithm
with cost $\lceil k/2\rceil$ and error $\frac{1}{2}-2^{-\Theta(k)}.$
In particular, for an arbitrary constant $k\geq1,$ the bounded-error
quantum query complexity is at most $\lceil k/2\rceil.$ (The original
version of~\cite{bansal-sinha20forrelation}, released concurrently
with our paper, had a poorer error parameter: $\frac{1}{2}-(\log n)^{-\Theta(k)}.$
But the authors of~\cite{bansal-sinha20forrelation} were able to
improve it several weeks later to match our error parameter, $\frac{1}{2}-2^{-\Theta(k)}.)$

The two approaches have incomparable strengths. To start with, Bansal
and Sinha~\cite{bansal-sinha20forrelation} prove their lower bound
for an \emph{explicit} function $f$ (namely, the forrelation and
rorrelation problems with a properly chosen gap parameter), as opposed
to the uniformly random choice of $f_{n,k,U}$ in this paper.

On the other hand, our analysis has the advantage of determining the
$\ell$-Fourier weight of decision trees. This result is of independent
interest beyond quantum computing, given the numerous recent applications
of Fourier weight to learning theory and pseudorandom generators.
We believe that our techniques may be relevant to other unresolved
questions on the Fourier spectrum of Boolean functions. The work in~\cite{bansal-sinha20forrelation},
by contrast, does not imply any improved bounds on Fourier weight.

Another strength of our analysis is methodological. The proof in~\cite{bansal-sinha20forrelation}
uses advanced analytic machinery, whereas our approach is elementary
and self-contained. Indeed, the only analytic fact used in this paper
and Tal~\cite{Tal19Query} is the p.d.f.~of the multivariate normal
distribution. With this simple toolkit, we obtain all the same optimal
quantum-classical separations for query complexity and communication
complexity as in~\cite{bansal-sinha20forrelation}.

\subsection{Follow-up work and future directions}

Of the work subsequent to
our paper, the most relevant result is due to Girish,
Tal, and Wu~\cite{GTW21parity}.  Those authors prove an
upper bound of $d^{\ell/2}\cdot O(\ell\log n)^\ell$ on
the $\ell$-Fourier weight of any \emph{parity} decision
tree of depth $d$ in $n$ variables.  
Plugging this bound into the machinery of 
Bansal and
Sinha~\cite{bansal-sinha20forrelation}, Girish et al.~obtain
a separation of $t$ versus 
$\tilde\Omega(n^{1-\frac1{2t}})$ for quantum query complexity
versus 
randomized \emph{parity} decision tree complexity, for
any constant $t\geq1.$ 
To compare their results with the
corresponding contributions of our work
(Corollary~\ref{cor:bounded-error-t-separation} and
Theorem~\ref{thm:MAIN-Fourier-bound}),
the \emph{parity} decision tree model of Girish et al.\ 
is more powerful than the standard randomized decision trees that we
consider. On the other hand, the Fourier weight bound
of~\cite{GTW21parity} deteriorates rapidly with $\ell$
and is not known to be tight beyond $\ell=O(1).$ Recall
that our Fourier weight bound in Theorem~\ref{thm:MAIN-Fourier-bound}
is essentially optimal for every $\ell.$ There are
methodological differences as well. For example, the
quantum-classical separation in~\cite{GTW21parity}
relies on the advanced machinery of Bansal and
Sinha~\cite{bansal-sinha20forrelation},
whereas our separation does not.

We close this section with a direction for future work.
In our separation of quantum versus classical query complexity
(Corollary~\ref{cor:bounded-error-t-separation}), the classical algorithm needs
$\tilde{\Omega}(n^{1-\frac{1}{2t}})$ queries to solve the problem with bounded
error, whereas the quantum algorithm makes precisely $t$ queries and succeeds
with probability $\frac{1}{2}+2^{-O(t)}.$ The same applies to the
quantum-classical query separation due to Bansal and
Sinha~\cite{bansal-sinha20forrelation} and the
follow-up separation of Girish et al.~\cite{GTW21parity}. In these results, the quantum 
algorithms conform to the bounded-error regime only for constant $t.$ A natural
open problem is to obtain an optimal separation in the bounded-error regime
for all $t=\omega(1).$

\section{Preliminaries}

\subsection{General notation}

There are two common arithmetic encodings for the Boolean values:
the traditional encoding $\text{false}\leftrightarrow0,\;\text{true}\leftrightarrow1,$
and the Fourier-motivated encoding $\text{false}\leftrightarrow1,\;\text{true}\leftrightarrow-1.$
Throughout this manuscript, we use the former encoding for the range
of a Boolean function and the latter for the domain. With this convention,
Boolean functions are mappings $\pomon\to\zoo$ for some~$n.$

We denote the empty string as usual by $\varepsilon.$ For an alphabet
$\Sigma$ and a natural number $n,$ we let $\Sigma^{\leq n}$ denote
the set of all strings over $\Sigma$ of length up to $n$, so that
$\Sigma^{\leq n}=\{\varepsilon\}\cup\Sigma\cup\Sigma^{2}\cup\cdots\cup\Sigma^{n}.$
For a string $v$ over a given alphabet, we let $|v|$ denote the
length of $v.$ For a set $S,$ we let $v|_{S}$ denote the substring
of $v$ indexed by the elements of $S.$ In other words, $v|_{S}=v_{i_{1}}v_{i_{2}}\cdots v_{i_{|S|}}$
where $i_{1}<i_{2}<\cdots<i_{|S|}$ are the elements of $S.$ In the
same spirit, we define $v_{\leq i}=v_{1}v_{2}\ldots v_{i}.$

The power set of a set $S$ is denoted by $\Pcal(S).$ For a set $S$
and a nonnegative integer $k,$ we let $\binom{S}{k}$ denote the
family of subsets of $S$ that have cardinality exactly~$k$:
\[
\binom{S}{k}=\{S'\subseteq S:|S'|=k\}.
\]
We further define
\[
\Pcal_{n,k}=\binom{\{1,2,\ldots,n\}}{k}=\{S\subseteq\{1,2,\ldots,n\}:|S|=k\}.
\]
The following well-known bound~\cite[Proposition~1.4]{jukna11extremal-2nd-edition}
is used in our proofs without further mention: 
\begin{align}
\left(\frac{n}{k}\right)^{k}\leq{n \choose k}\leq\left(\frac{\e n}{k}\right)^{k}, &  & k=1,2,\dots,n,\label{eq:entropy-bound-binomial}
\end{align}
where $\e=2.7182\ldots$ denotes Euler's number.

We adopt the standard notation $\NN=\{0,1,2,3,\ldots\}$ and $\ZZ^{+}=\{1,2,3,\ldots\}$
for the sets of natural numbers and positive integers, respectively.
We adopt the extended real number system $\Re\cup\{-\infty,\infty\}$
in all calculations. The functions $\ln x$ and $\log x$ stand for
the natural logarithm of $x$ and the logarithm of $x$ to base $2,$
respectively. To avoid excessive use of parentheses, we follow the
notational convention that $\ln a_{1}a_{2}\ldots a_{k}=\ln(a_{1}a_{2}\ldots a_{k})$
for any factors $a_{1},a_{2},\ldots,a_{k}.$ The binary entropy function
$H\colon[0,1]\to[0,1]$ is given by
\[
H(x)=x\log\frac{1}{x}+(1-x)\log\frac{1}{1-x}.
\]
Basic calculus reveals that
\begin{equation}
H(x)\leq1-\frac{2}{\ln2}\left(x-\frac{1}{2}\right)^{2}.\label{eq:entropy-upper-bound}
\end{equation}
For nonempty sets $A,B\subseteq\Re,$ we write $A<B$ to mean that
$a<b$ for all $a\in A,\,b\in B.$ It is clear that this relation
is a partial order on nonempty subsets of $\Re.$ We use the standard
definition of the sign function:
\[
\sign x=\begin{cases}
-1 & \text{if }x<0,\\
0 & \text{if }x=0,\\
1 & \text{if }x>0.
\end{cases}
\]
For a finite set $X$, we let $\Re^{X}$ denote the family of real-valued
functions on $X.$ For $f,g\in\Re^{X},$ we let $f\cdot g\in\Re^{X}$
denote the pointwise product of $f$ and $g,$ with $(f\cdot g)(x)=f(x)g(x).$
We use the standard inner product $\langle f,g\rangle=\sum_{x\in X}f(x)g(x).$

\subsection{Fourier transform}
\label{sec:fourier}

Consider the real vector space of functions $\pomon\to\Re.$ For $S\subseteq\oneton,$
define $\chi_{S}\colon\pomon\to\pomo$ by $\chi_{S}(x)=\prod_{i\in S}x_{i}.$
Then 
\[
\langle\chi_{S},\chi_{T}\rangle=\begin{cases}
2^{n} & \text{if }S=T,\\
0 & \text{otherwise.}
\end{cases}
\]
Thus, $\{\chi_{S}\}_{S\subseteq\oneton}$ is an orthogonal basis for
the vector space in question. In particular, every function $\phi\colon\pomon\to\Re$
has a unique representation of the form 
\begin{align*}
\phi=\sum_{S\subseteq\oneton}\hat{\phi}(S)\chi_{S}
\end{align*}
for some reals $\hat{\phi}(S),$ where by orthogonality $\hat{\phi}(S)=2^{-n}\langle\phi,\chi_{S}\rangle$.
The reals $\hat{\phi}(S)$ are called the \emph{Fourier coefficients
of $\phi,$} and the mapping $\phi\mapsto\hat{\phi}$ is the \emph{Fourier
transform of $\phi.$} Put another way, every function $\phi:\pomon\to\Re$
has a unique representation as a multilinear polynomial 
\begin{equation}
\phi(x)=\sum_{S\subseteq\oneton}\hat{\phi}(S)\prod_{i\in S}x_{i},\label{eq:phi-multilinear}
\end{equation}
where the real numbers $\hat{\phi}(S)$ are the Fourier coefficients
of $f.$ The \emph{order} of a Fourier coefficient $\hat{\phi}(S)$
is the cardinality $|S|$.

For $k=0,1,2,\ldots,n,$ we introduce the linear operator $L_{k}\colon\Re^{\pomo^{n}}\to\Re^{\pomo^{n}}$
that sends a function $\phi\colon\pomon\to\Re$ to the function $L_{k}\phi\colon\pomon\to\Re$
given by 
\[
(L_{k}\phi)(x)=\sum_{S\in\Pcal_{n,k}}\hat{\phi}(S)\chi_{S}(x).
\]
We refer to $L_{k}\phi$ as the \emph{degree-$k$ homogeneous part
of $\phi.$}

For any polynomial $p\in\Re[x_{1},x_{2},\ldots,x_{n}],$ we let $\norm p$
denote the sum of the absolute values of the coefficients of $p.$
One easily verifies the well-known fact that $\norm\cdot$ is a norm
on the polynomial ring $\Re[x_{1},x_{2},\ldots,x_{n}].$ We identify
a function $\phi\colon\pomon\to\Re$ with its unique representation~(\ref{eq:phi-multilinear})
as a multilinear polynomial, to the effect that 
\[
\norm\phi=\sum_{S\subseteq\oneton}|\hat{\phi}(S)|
\]
is the sum of the absolute values of the Fourier coefficients of $\phi.$
\begin{proposition}
\label{prop:fourier-norm-convex}For any functions $\phi,\psi\colon\pomon\to\Re$
and reals $a,b,$
\begin{align*}
\norm{a\phi+b\psi}\leq|a|\,\norm\phi+|b|\,\norm\psi.
\end{align*}
\end{proposition}

\begin{proof}
We have
\begin{align*}
\norm{a\phi+b\psi} & =\sum_{S\subseteq\oneton}|a\hat{\phi}(S)+b\hat{\psi}(S)|\\
 & \leq|a|\sum_{S\subseteq\oneton}|\hat{\phi}(S)|+|b|\sum_{S\subseteq\oneton}|\hat{\psi}(S)|\\
 & =|a|\,\norm\phi+|b|\,\norm\psi,
\end{align*}
where the first step uses the linearity of the Fourier transform.
\end{proof}
We also note the following submultiplicative property.
\begin{proposition}
\label{prop:fourier-norm-submult}For any functions $\phi,\psi\colon\pomon\to\Re,$
\begin{align*}
\norm{\phi\cdot\psi}\leq\norm\phi\,\norm\psi.
\end{align*}
\end{proposition}

\begin{proof}
We have
\begin{align*}
\phi\cdot\psi & =\left(\sum_{S\subseteq\oneton}\hat{\phi}(S)\chi_{S}\right)\left(\sum_{T\subseteq\oneton}\hat{\psi}(T)\chi_{T}\right)\\
 & =\sum_{S,T\subseteq\oneton}\hat{\phi}(S)\hat{\psi}(T)\chi_{(S\setminus T)\cup(T\setminus S).}
\end{align*}
Applying Proposition~\ref{prop:fourier-norm-convex},
\begin{align*}
\norm{\phi\cdot\psi} & \leq\sum_{S,T\subseteq\oneton}|\hat{\phi}(S)|\,|\hat{\psi}(T)|.
\end{align*}
The right-hand side of this inequality is clearly $\norm\phi\,\norm\psi.$
\end{proof}
We will frequently use the norm $\norm\cdot$ in conjunction with
the operator $L_{k}$ to refer to the sum of the absolute values of
the Fourier coefficients of a given order $k$:
\[
\norm{L_{k}\,\phi}=\sum_{S\in\Pcal_{n,k}}|\hat{\phi}(S)|.
\]

\subsection{Generalized decision trees}
\label{sec:gen-d-t}

Throughout this manuscript, we assume decision trees to be perfect
binary trees, with each internal node having two children and all
leaves having the same depth. This convention is without loss of generality
since a decision tree computing a given function $f$ can be made
into a perfect binary tree for $f$ of the same depth, by querying
dummy variables as necessary. We denote the variables of a decision
tree by $x_{1},x_{2},\ldots,x_{n}\in\pomo,$ and identify the vertices
of a decision tree in the natural manner with strings in $\pomo^{*}.$
Thus, $\varepsilon$ denotes the root of the tree, and a string $v\in\pomo^{k}$
denotes the vertex at depth $k$ reached from the root by following
the path $v_{1}v_{2}\ldots v_{k}$. Formally, a \emph{decision tree}
of depth $d$ in Boolean variables $x_{1},x_{2},\ldots,x_{n}\in\pomo$
is a function $T$ on $\pomo^{\leq d}$ with the following two properties.
\begin{enumerate}
\item One has $T(v)\in\{1,2,\ldots,n\}$ for every $v\in\pomo^{\leq d-1}$,
with the interpretation that $T(v)$ is the index of the variable
queried at the internal node found by following the path $v=v_{1}v_{2}v_{3}\ldots$
from the root of the decision tree. We note that a variable cannot
be queried twice on the same path, and therefore the $d$ numbers
$T(\varepsilon),T(v_{1}),T(v_{1}v_{2}),\ldots,T(v_{1}v_{2}\ldots v_{d-1})$
are pairwise distinct for every $v\in\pomo^{d-1}$.
\item One has $T(v)\in\Re[x_{1},x_{2},\ldots,x_{n}]$ for every $v\in\pomo^{d}$,
with the interpretation that $T(v)$ is the label of the leaf reached
by following the path $v=v_{1}v_{2}\ldots v_{d}$ from the root of
the tree. Thus, every leaf is labeled with a real-valued polynomial
in the input variables $x_{1},x_{2},\ldots,x_{n}$. At a given leaf
$v\in\pomo^{d},$ the variables $x_{T(\varepsilon)},x_{T(v_{1})},\ldots,x_{T(v_{1}v_{2}\ldots v_{d-1})}$
have been queried and therefore have fixed values. For this reason,
we require $T(v)$ to be a real polynomial in variables other than
$x_{T(\varepsilon)},x_{T(v_{1})},\ldots,x_{T(v_{1}v_{2}\ldots v_{d-1})}$.
We refer to a leaf $v\in\pomo^{d}$ as a \emph{nonzero leaf} if $T(v)$
is not the zero polynomial. While we formally allow arbitrary real
polynomials, the identity $x_{i}^{2}=x_{i}$ effectively forces $T(v)$
for each $v\in\pomo^{d}$ to be multilinear.
\end{enumerate}
Our formalism generalizes the traditional notion of a decision tree,
where the leaf labels are restricted to the Boolean constants $0$
and $1$.
\begin{proposition}
\label{prop:tree-computation}Let $T$ be a given decision tree of
depth $d.$ Then the function $f\colon\pomon\to\Re$ computed by $T$
is given by
\begin{equation}
f(x)=\sum_{v\in\pomo^{d}}\;T(v)\cdot\prod_{i=1}^{d}\frac{1+v_{i}x_{T(v_{1}v_{2}\ldots v_{i-1})}}{2}.\label{eq:f-T}
\end{equation}
\end{proposition}

\noindent We emphasize that $T(v)$ in this expression is a polynomial
in $x_{1},x_{2},\ldots,x_{n}$ and not necessarily a constant value.
In fact, the norm $\norm{T(v)}$ for leaves $v$ is a prominent quantity
in this paper.
\begin{proof}
For an input $x\in\pomon$ and a leaf $v\in\pomo^{d},$ the product
\[
\prod_{i=1}^{d}\frac{1+v_{i}x_{T(v_{1}v_{2}\ldots v_{i-1})}}{2}
\]
evaluates to $1$ if the input $x$ reaches the leaf $v$ in $T$,
and evaluates to $0$ otherwise. Recall that any given input $x$
reaches precisely one leaf $v,$ and the output of the tree on $x$
is defined to be the corresponding polynomial $T(v)\in\Re[x_{1},x_{2},\ldots,x_{n}]$
evaluated at $x.$ Thus, (\ref{eq:f-T}) evaluates to $T(v)$ where
$v$ is the leaf reached by $x.$
\end{proof}
For a decision tree $T$ of depth $d,$ we let $\dns(T)$ denote the
fraction of leaves in $T$ with nonzero labels:
\[
\dns(T)=\Prob_{v\in\pomo^{d}}[T(v)\ne0].
\]
We refer to this quantity as the \emph{density} of $T$. Another important
complexity measure is the \emph{degree of $T,$} denoted $\deg(T)$
and defined as the maximum of the degrees of the polynomials $T(v)\in\Re[x_{1},x_{2},\ldots,x_{n}]$
for $v\in\pomo^{d}.$ Recall that the zero polynomial $0$ is considered
to have degree $-\infty.$ For an internal node $v\in\pomo^{\leq d-1},$
we let $T_{v}$ denote the subtree of $T$ rooted at $v.$ Thus, $T_{v}$
is the tree of depth $d-|v|$ given by $T_{v}(u)=T(vu)$ for all $u\in\pomo^{\leq d-|v|}$.
The following fact is straightforward and well-known.
\begin{fact}
\label{fact:dns-prob-nonzero}Let $T$ be a given decision tree of
degree at most~$0$. Let $f\colon\pomon\to\Re$ be the function computed
by $T$. Then 
\[
\Prob_{x\in\pomon}[f(x)\ne0]=\dns(T).
\]
\end{fact}

\begin{proof}
Let $d$ be the depth of $T$. Since $T$ is a perfect binary tree,
the fraction of inputs $x\in\pomon$ that reach any given leaf of
$T$ is exactly $2^{-d}$. Therefore, the probability that a random
input $x\in\pomon$ reaches a leaf with a nonzero label is precisely
the fraction of leaves with nonzero labels, which is by definition
$\dns(T)$.
\end{proof}
We will be working with special classes of trees described by several
parameters. Specifically, we let $\Tcal(n,d,p,k)$ denote the set
of all trees in $n$ Boolean variables $x_{1},x_{2},\ldots,x_{n}\in\pomo$
of depth $d$ and density $p$ such that for every leaf $v\in\pomo^{d}$,
the label $T(v)$ is either the zero polynomial $0$ or a homogeneous
multilinear polynomial of degree $k$. We further define $\Tcal^{*}(n,d,p,k)$
to be the set of all trees $T\in\Tcal(n,d,p,k)$ that have the additional
property that $T(v)\in\{0\}\cup\{\pm\prod_{i\in S}x_{i}\colon S\in\Pcal_{n,k}\}$
for every leaf $v\in\pomo^{d}.$ Thus, every nonzero leaf in a tree
$T\in\Tcal^{*}(n,d,p,k)$ is labeled with a signed monomial of degree
$k.$

The Fourier spectrum of decision trees has been studied in several
works, as discussed in the introduction. We will need the following
special case of a result due to Tal~\cite[Theorem~7.5]{Tal19Query}.
\begin{theorem}[Tal]
\label{thm:tal-depth-2} Let $f\colon\pomon\to\{-1,0,1\}$ be given,
$f\not\equiv0$. Define $p=\Prob_{x\in\pomon}[f(x)\ne0].$ Suppose
that $f$ can be computed by a depth-$d$ decision tree. Then
\begin{align*}
\norm{L_{1}\,f} & \leq\binom{d}{1}^{1/2}Cp\sqrt{\ln\frac{e}{p}},\\
\norm{L_{2}\,f} & \leq\binom{d}{2}^{1/2}C^{2}p\sqrt{\ln\frac{e}{p}}\sqrt{\ln\frac{en}{p}},
\end{align*}
where $C\geq1$ is an absolute constant.
\end{theorem}

\noindent Tal states his result for functions $f\colon\pomon\to\zoo$
rather than $f\colon\pomon\to\{-1,0,1\}.$ But Theorem~\ref{thm:tal-depth-2}
follows immediately by writing $f=f^{+}-f^{-}$, where $f^{+},f^{-}\colon\pomon\to\zoo$
are the positive and negative parts of $f,$ and applying Tal's result
separately to $f^{+}$ and $f^{-}.$

\section{\label{sec:Elementary-set-families}Elementary set families}

As explained in the introduction, we obtain our Fourier weight bound
by combining the Fourier coefficients of a decision tree into well-structured
groups and bounding the sum of the absolute values in each group.
In this section, we lay the combinatorial groundwork for this result
by proving that $\Pcal_{n,k}$ can be efficiently partitioned into
what we call ``elementary families.'' We start in Section~\ref{subsec:A-binomial-recurrence}
with some technical calculations. Section~\ref{subsec:The-partition-measure}
formally defines elementary families and studies the associated complexity
measure for representing general families as the disjoint union of
elementary parts. Finally, Section~\ref{subsec:An-efficient-partition}
proves that our family of interest $\Pcal_{n,k}$ has an efficient
partition of this form.

\subsection{\label{subsec:A-binomial-recurrence}A binomial recurrence}

Our starting point is a technical calculation related to the entropy
function.
\begin{lemma}
\label{lem:sqrt-double-binomial-sum}There is an absolute constant
$c\geq1$ such that for all integers $k\geq1,$
\[
\sum_{i=1}^{k-1}\left(\frac{k}{i}\right)^{i/2}\left(\frac{k}{k-i}\right)^{(k-i)/2}\frac{1}{\sqrt{i(k-i)}}\leq c\;\sqrt{\frac{2^{k}}{k}}.
\]
\end{lemma}

\begin{proof}
To begin with,
\begin{align}
\sum_{i=1}^{k-1}\left(\frac{k}{i}\right)^{i/2} & \left(\frac{k}{k-i}\right)^{(k-i)/2}\frac{1}{\sqrt{i(k-i)}}\nonumber \\
 & =\sum_{i=1}^{k-1}\frac{2^{H(i/k)\cdot k/2}}{\sqrt{i(k-i)}}\nonumber \\
 & \leq2^{k/2}\sum_{i=1}^{k-1}\exp\left(-k\left(\frac{i}{k}-\frac{1}{2}\right)^{2}\right)\cdot\frac{1}{\sqrt{i(k-i)}},\label{eq:k-induction-start}
\end{align}
where the last step uses~(\ref{eq:entropy-upper-bound}). Continuing,
\begin{align}
\sum_{i=1}^{\lceil k/4\rceil-1}\exp\left(-k\left(\frac{i}{k}-\frac{1}{2}\right)^{2}\right)\frac{1}{\sqrt{i(k-i)}} & \leq\sum_{i=1}^{\lceil k/4\rceil-1}\exp\left(-k\left(\frac{i}{k}-\frac{1}{2}\right)^{2}\right)\nonumber \\
 & \leq\sum_{i=1}^{\lceil k/4\rceil-1}\e^{-k/16}\nonumber \\
 & <\frac{k\e^{-k/16}}{4}.
\end{align}
Symmetrically,
\begin{equation}
\sum_{i=\lfloor3k/4\rfloor+1}^{k-1}\exp\left(-k\left(\frac{i}{k}-\frac{1}{2}\right)^{2}\right)\frac{1}{\sqrt{i(k-i)}}<\frac{k\e^{-k/16}}{4}.
\end{equation}
Finally,
\begin{align}
\sum_{i=\lceil k/4\rceil}^{\lfloor3k/4\rfloor}\exp & \left(-k\left(\frac{i}{k}-\frac{1}{2}\right)^{2}\right)\frac{1}{\sqrt{i(k-i)}}\nonumber \\
 & \qquad\leq\frac{4}{\sqrt{3}k}\sum_{i=\lceil k/4\rceil}^{\lfloor3k/4\rfloor}\exp\left(-k\left(\frac{i}{k}-\frac{1}{2}\right)^{2}\right)\nonumber \\
 & \qquad\leq\frac{4}{\sqrt{3}k}\sum_{i=-\infty}^{\infty}\exp\left(-k\left(\frac{i}{k}-\frac{1}{2}\right)^{2}\right)\nonumber \\
 & \qquad\leq\frac{4}{\sqrt{3}k}+\frac{4}{\sqrt{3}k}\int_{-\infty}^{\infty}\exp\left(-k\left(\frac{x}{k}-\frac{1}{2}\right)^{2}\right)dx\nonumber \\
 & \qquad=\frac{4}{\sqrt{3}k}+\frac{4\sqrt{\pi}}{\sqrt{3k}}.\label{eq:k-induction-end}
\end{align}
Combining (\ref{eq:k-induction-start})\textendash (\ref{eq:k-induction-end}),
we conclude that
\begin{align*}
\sum_{i=1}^{k-1}\left(\frac{k}{i}\right)^{i/2}\left(\frac{k}{k-i}\right)^{(k-i)/2}\frac{1}{\sqrt{i(k-i)}} & \leq2^{k/2}\left(\frac{k\e^{-k/16}}{2}+\frac{4}{\sqrt{3}k}+\frac{4\sqrt{\pi}}{\sqrt{3k}}\right).
\end{align*}
This settles the lemma for a large enough absolute constant $c\geq1.$
\end{proof}
As an application of the previous lemma, we proceed to solve a key
recurrence that we will need to study $\Pcal_{n,k}$.
\begin{theorem}
\label{thm:N-n-k}Let $N\colon\{1,2,4,8,16,\ldots\}\times\ZZ^{+}\to[0,\infty)$
be any function that satisfies 
\begin{align*}
N(n,k) & \leq\binom{n}{k}^{1/2} & \qquad & \text{if }\min\{n,k\}\leq2,\\
N(n,k) & \leq2N\left(\frac{n}{2},k\right)+\sum_{i=1}^{k-1}N\!\left(\frac{n}{2},i\right)N\!\left(\frac{n}{2},k-i\right) & \qquad & \text{if }\min\{n,k\}>2.
\end{align*}
Let $c\geq1$ be the absolute constant from Lemma~\emph{\ref{lem:sqrt-double-binomial-sum}}.
Then for all $n,k,$
\begin{align}
N(n,k) & \leq\frac{(2+\sqrt{2})^{k-1}c^{k-1}}{\sqrt{k}}\left(\frac{n}{k}\right)^{k/2}.\label{eq:N-n-k-inductive-hypothesis}
\end{align}
\end{theorem}

\begin{proof}
The proof of (\ref{eq:N-n-k-inductive-hypothesis}) is by induction
on the pair $(n,k)\in\{1,2,4,8,16,\ldots\}\times\mathbb{Z}^{+}.$
For $\min\{n,k\}\leq2,$ the claimed bound~(\ref{eq:N-n-k-inductive-hypothesis})
is a weakening of $N(n,k)\leq\binom{n}{k}^{1/2}$. This establishes
the base case. For the inductive step, fix any $n\in\{4,8,16,32,\ldots\}$
and $k\geq3.$ Abbreviate $\alpha=2+\sqrt{2}.$ Then
\begin{align*}
N(n,k) & \leq2N\!\left(\frac{n}{2},k\right)+\sum_{i=1}^{k-1}N\!\left(\frac{n}{2},i\right)N\!\left(\frac{n}{2},k-i\right)\\
 & \leq2\cdot\frac{(\alpha c)^{k-1}}{\sqrt{k}}\left(\frac{n}{2k}\right)^{k/2}\\
 & \qquad\qquad+\sum_{i=1}^{k-1}\frac{(\alpha c)^{i-1}}{\sqrt{i}}\left(\frac{n}{2i}\right)^{i/2}\cdot\frac{(\alpha c)^{k-i-1}}{\sqrt{k-i}}\left(\frac{n}{2(k-i)}\right)^{(k-i)/2}\\
 & =2\cdot\frac{(\alpha c)^{k-1}}{\sqrt{k}}\left(\frac{n}{2k}\right)^{k/2}\\
 & \qquad\qquad+(\alpha c)^{k-2}\left(\frac{n}{2k}\right)^{k/2}\sum_{i=1}^{k-1}\frac{1}{\sqrt{i(k-i)}}\left(\frac{k}{i}\right)^{i/2}\left(\frac{k}{k-i}\right)^{(k-i)/2}\\
 & \leq2\cdot\frac{(\alpha c)^{k-1}}{\sqrt{k}}\left(\frac{n}{2k}\right)^{k/2}+\frac{(\alpha c)^{k-2}c}{\sqrt{k}}\left(\frac{n}{k}\right)^{k/2}\\
 & \leq\frac{1}{\sqrt{2}}\cdot\frac{(\alpha c)^{k-1}}{\sqrt{k}}\left(\frac{n}{k}\right)^{k/2}+\frac{(\alpha c)^{k-2}c}{\sqrt{k}}\left(\frac{n}{k}\right)^{k/2}\\
 & =\frac{(\alpha c)^{k-1}}{\sqrt{k}}\left(\frac{n}{k}\right)^{k/2},
\end{align*}
where the second step applies the inductive hypothesis; the fourth
step appeals to Lemma~\ref{lem:sqrt-double-binomial-sum}; and the
fifth step uses $k\geq3.$ This completes the inductive step and thereby
settles~(\ref{eq:N-n-k-inductive-hypothesis}).
\end{proof}

\subsection{\label{subsec:The-partition-measure}The partition measure}

For set families $\Acal,\Bcal\subseteq\Pcal(\ZZ)$, we define $\Acal\ast\Bcal=\{A\cup B:A\in\Acal,B\in\Bcal\}.$
We collect basic properties of this operation in the proposition below.
\begin{proposition}
\label{prop:ast-properties}Let $\Acal,\Bcal,\Ccal\subseteq\Pcal(\ZZ)$
be given. Then:
\begin{enumerate}
\item $\Acal*\varnothing=\varnothing*\Acal=\varnothing;$
\item \label{enu:ast-properties-idenpotence}$\Acal*\{\varnothing\}=\{\varnothing\}*\Acal=\Acal;$
\item $(\Acal\ast\Bcal)\ast\Ccal=\Acal\ast(\Bcal*\Ccal);$
\item $\Acal\ast\Bcal=\Bcal\ast\Acal;$
\item $(\Acal\cup\Bcal)\ast\Ccal=(\Acal\ast\Ccal)\cup(\Bcal\ast\Ccal).$
\end{enumerate}
\end{proposition}

\begin{proof}
All properties are immediate from the definition of the $\ast$ operation.
\end{proof}
We define an \emph{integer interval} to be any finite set whose elements
are consecutive integers, namely, $\{i,i+1,i+2,\ldots,j\}$ for some
$i,j\in\ZZ.$ As a special case, this includes the empty interval
$\varnothing$. An \emph{elementary family} is any family of the form
\begin{equation}
\Ecal=\binom{I_{1}}{k_{1}}\ast\binom{I_{2}}{k_{2}}\ast\cdots\ast\binom{I_{\ell}}{k_{\ell}},\label{eq:Ecal-defined}
\end{equation}
where $\ell$ is a positive integer, $I_{1},I_{2},\ldots,I_{\ell}$
are pairwise disjoint integer intervals, and $k_{1},k_{2},\ldots,k_{\ell}\in\{0,1,2\}.$
Trivial examples of elementary families are $\binom{\varnothing}{0}=\{\varnothing\}$
and $\binom{\varnothing}{1}=\varnothing.$ Another example of an elementary
family is the singleton family $\{A\}$ for any nonempty finite set
$A\subseteq\ZZ,$ using $\{A\}=\binom{\{a_{1}\}}{1}*\binom{\{a_{2}\}}{1}*\cdots*\binom{\{a_{\ell}\}}{1}$
where $a_{1}<a_{2}<\cdots<a_{\ell}$ are the distinct elements of
$A.$ We now define a partition measure that captures how efficiently
a family can be partitioned into elementary families.
\begin{definition}[Partition measure $\pi$]
For any family $\Acal\subseteq\Pcal(\oneton),$ define $\pi(\Acal)$
to be the minimum
\begin{equation}
\sum_{i=1}^{N}|\Ecal_{i}|^{1/2}\label{eq:minimize-partition}
\end{equation}
over all integers $N$ and all elementary families $\Ecal_{1},\Ecal_{2},\ldots,\Ecal_{N}$
that are pairwise disjoint and satisfy $\Ecal_{1}\cup\Ecal_{2}\cup\cdots\cup\Ecal_{N}=\Acal.$
\end{definition}

Straight from the definition, 
\begin{align*}
 & \pi(\varnothing)=0,\\
 & \pi(\{\varnothing\})=1.
\end{align*}
More generally,
\begin{equation}
|\Acal|^{1/2}\leq\pi(\Acal)\leq|\Acal|\label{eq:pi-trivial-bounds}
\end{equation}
for every $\Acal\subseteq\Pcal(\oneton)$. The upper bound here corresponds
to the trivial partition $\Acal=\bigcup_{A\in\Acal}\{A\}$. The lower
bound holds because~(\ref{eq:minimize-partition}) is no smaller
than $(\sum|\Ecal_{i}|)^{1/2}=|\Acal|^{1/2}.$ The following four
lemmas will be useful to us in analyzing the partition measure for
families of interest.
\begin{lemma}
\label{lem:ast-partition}Let $\Acal,\Bcal\subseteq\Pcal(\oneton)$
be given with $\Acal\cap\Bcal=\varnothing.$ Then
\[
\pi(\Acal\cup\Bcal)\leq\pi(\Acal)+\pi(\Bcal).
\]
\end{lemma}

\begin{proof}
If $\Acal=\varnothing$ or $\Bcal=\varnothing,$ the claim is trivial.
In the complementary case, let $\Acal=\Ecal_{1}\cup\cdots\cup\Ecal_{N}$
and $\Bcal=\Ecal'_{1}\cup\cdots\cup\Ecal_{N'}'$ be partitions of
$\Acal$ and $\Bcal$, respectively, into elementary families. Then
$\Acal\cup\Bcal=(\Ecal_{1}\cup\cdots\cup\Ecal_{N})\cup(\Ecal'_{1}\cup\cdots\cup\Ecal_{N'}')$
is a partition of $\Acal\cup\Bcal$ into elementary families.
\end{proof}
\begin{lemma}
\label{lem:ast-product}Let $\Acal\subseteq\Pcal(\{1,2,\ldots,m\})$
and $\Bcal\subseteq\Pcal(\{m+1,m+2,\ldots,n\})$ be given, for some
$1\leq m<n.$ Then 
\[
\pi(\Acal\ast\Bcal)\leq\pi(\Acal)\,\pi(\Bcal).
\]
\end{lemma}

\begin{proof}
If $\Acal=\varnothing$ or $\Bcal=\varnothing,$ we have $\Acal\ast\Bcal=\varnothing$
by Proposition~\ref{prop:ast-properties} and therefore $\pi(\Acal\ast\Bcal)=0.$
In the complementary case, let $\Acal=\Ecal_{1}\cup\cdots\cup\Ecal_{N}$
and $\Bcal=\Ecal'_{1}\cup\cdots\cup\Ecal_{N'}'$ be partitions of
$\Acal$ and $\Bcal$, respectively, into elementary families for
which $\pi(\Acal)$ and $\pi(\Bcal)$ are achieved. Then 
\begin{equation}
\Acal\ast\Bcal=\left(\bigcup_{i=1}^{N}\Ecal_{i}\right)\ast\Bcal=\bigcup_{i=1}^{N}(\Ecal_{i}\ast\Bcal)=\bigcup_{i=1}^{N}\bigcup_{j=1}^{N'}(\Ecal_{i}\ast\Ecal_{j}'),\label{eq:A-B-decomp}
\end{equation}
where the last two steps use the distributivity and commutativity
properties in Proposition~\ref{prop:ast-properties}. For any elementary
families $\Ecal_{i}\subseteq\Pcal(\{1,2,\ldots,m\})$ and $\Ecal'_{j}\subseteq\Pcal(\{m+1,m+2,\ldots,n\}),$
the family $\Ecal_{i}\ast\Ecal'_{j}\subseteq\Pcal(\{1,2,\ldots,n\})$
is also elementary, with $|\Ecal_{i}\ast\Ecal_{j}'|=|\Ecal_{i}|\;|\Ecal_{j}'|$.
Since all unions in~(\ref{eq:A-B-decomp}) are disjoint, we obtain
\begin{align*}
\pi(\Acal\ast\Bcal) & \leq\sum_{i=1}^{N}\sum_{j=1}^{N'}|\Ecal_{i}\ast\Ecal_{j}'|^{1/2}=\sum_{i=1}^{N}\sum_{j=1}^{N'}|\Ecal_{i}|^{1/2}|\Ecal_{j}'|^{1/2}=\pi(\Acal)\pi(\Bcal).
\end{align*}
\end{proof}

For a set $A\subseteq\ZZ$ and an integer $x,$ we define $A+x=\{a+x:a\in A\}.$
Analogously, for a family $\Acal\subseteq\Pcal(\ZZ),$ we define $\Acal+x=\{A+x:A\in\Acal\}.$
As one would expect, the partition measure is invariant under translation
by an integer.
\begin{lemma}
\label{lem:partition-translate}Let $\Acal\subseteq\Pcal(\oneton)$
be given. Then for all $x\in\NN,$
\[
\pi(\Acal)=\pi(\Acal+x).
\]
\end{lemma}

\begin{proof}
Consider an elementary family $\Ecal$ of the form~(\ref{eq:Ecal-defined}),
where $I_{1},I_{2},\ldots,I_{\ell}$ are pairwise disjoint integer
intervals and $k_{1},k_{2},\ldots,k_{\ell}\in\{0,1,2\}.$ Then
\[
\Ecal+x=\binom{I_{1}+x}{k_{1}}\ast\binom{I_{2}+x}{k_{2}}\ast\cdots\ast\binom{I_{\ell}+x}{k_{\ell}}
\]
is also an elementary family because the translated integer intervals
$I_{1}+x,I_{2}+x,\ldots,I_{\ell}+x$ are pairwise disjoint. Thus,
any partition $\Acal=\bigcup_{i=1}^{N}\Ecal_{i}$ into elementary
families gives an analogous partition $\Acal+x=\bigcup_{i=1}^{N}(\Ecal_{i}+x)$
into elementary families, with $|\Ecal_{i}+x|=|\Ecal_{i}|$ for all
$i.$
\end{proof}
In general, $\Acal\subseteq\Bcal$ does not imply $\pi(\Acal)\leq\pi(\Bcal)$.
However, $\pi$ enjoys the following monotonicity property.
\begin{lemma}
\label{lem:partition-monotone}For any positive integers $n,m,k$
with $n\leq m,$ 
\[
\pi(\Pcal_{n,k})\leq\pi(\Pcal_{m,k}).
\]
\end{lemma}

\begin{proof}
Consider an elementary family $\Ecal$ of the form~(\ref{eq:Ecal-defined}),
where $I_{1},I_{2},\ldots,I_{\ell}$ are pairwise disjoint integer
intervals and $k_{1},k_{2},\ldots,k_{\ell}\in\{0,1,2\}.$ Then
\[
\Ecal\cap\Pcal(\{1,2,\ldots,n\})=\binom{I_{1}\cap\{1,2,\ldots,n\}}{k_{1}}\ast\cdots\ast\binom{I_{\ell}\cap\{1,2,\ldots,n\}}{k_{\ell}}
\]
is also an elementary family because the integer intervals $I_{j}\cap\{1,2,\ldots,n\}$
for $j=1,2,\ldots,\ell$ are pairwise disjoint. Thus, any partition
$\Pcal_{m,k}=\bigcup_{i=1}^{N}\Ecal_{i}$ into elementary families
gives an analogous partition for $\Pcal_{n,k}$:
\begin{align*}
\Pcal_{n,k} & =\Pcal_{m,k}\cap\Pcal(\{1,2,\ldots,n\})\\
 & =\bigcup_{i=1}^{N}\Ecal_{i}\cap\Pcal(\{1,2,\ldots,n\}).
\end{align*}
Moreover, the elementary families in the new partition obey $|\Ecal_{i}\cap\Pcal(\{1,2,\ldots,n\})|\leq|\Ecal_{i}|$
for all $i.$
\end{proof}

\subsection{\label{subsec:An-efficient-partition}An efficient partition for
$\Pcal_{n,k}$}

Our analysis of the Fourier spectrum of decision trees relies on the
partition measure of the family $\Pcal_{n,k}.$ Recall from~(\ref{eq:pi-trivial-bounds})
that 
\[
\pi(\Pcal_{n,k})\geq\binom{n}{k}^{1/2}.
\]
We will now prove that this lower bound is tight up to a factor of
$2^{O(k)},$ by combining Lemmas~\ref{lem:ast-partition}\textendash \ref{lem:partition-monotone}
with the recurrence solved in Theorem~\ref{thm:N-n-k}.
\begin{theorem}
\label{thm:pi-n-choose-k}Let $c\geq1$ be the absolute constant from
Lemma~\emph{\ref{lem:sqrt-double-binomial-sum}}. Then for all positive
integers $n$ and $k,$
\begin{equation}
\pi(\Pcal_{n,k})\leq\frac{(2+\sqrt{2})^{k-1}c^{k-1}}{\sqrt{k}}\left(\frac{2n}{k}\right)^{k/2}.\label{eq:pi-n-choose-k}
\end{equation}
\end{theorem}

\begin{proof}
We first treat the case when $n$ is a power of $2.$ If $k\leq2,$
the family $\Pcal_{n,k}$ is elementary to start with. As a result,
\begin{align}
\pi(\Pcal_{n,k}) & \leq\binom{n}{k}^{1/2}, &  & k\leq2.\label{eq:pi-base-k}
\end{align}
If $n\leq2,$ the family $\Pcal_{n,k}$ is empty unless $k\leq2$.
Therefore, again
\begin{align}
\pi(\Pcal_{n,k}) & \leq\binom{n}{k}^{1/2}, &  & n\leq2.\label{eq:pi-base-n}
\end{align}
For $n,k\geq3,$ we have 
\begin{align}
\pi(\Pcal_{n,k}) & =\pi\left(\bigcup_{i=0}^{k}\left(\binom{\{1,2,\ldots,n/2\}}{i}\ast\binom{\{n/2+1,n/2+2,\ldots,n\}}{k-i}\right)\right)\nonumber \\
 & \leq\sum_{i=0}^{k}\pi\left(\binom{\{1,2,\ldots,n/2\}}{i}\ast\binom{\{n/2+1,n/2+2,\ldots,n\}}{k-i}\right)\nonumber \\
 & \leq\sum_{i=0}^{k}\pi\left(\binom{\{1,2,\ldots,n/2\}}{i}\right)\pi\left(\binom{\{n/2+1,n/2+2,\ldots,n\}}{k-i}\right)\nonumber \\
 & =\sum_{i=0}^{k}\pi(\Pcal_{n/2,i})\,\pi\left(\Pcal_{n/2,k-i}+\frac{n}{2}\right)\nonumber \\
 & =\sum_{i=0}^{k}\pi(\Pcal_{n/2,i})\,\pi(\Pcal_{n/2,k-i})\nonumber \\
 & =2\pi(\Pcal_{n/2,k})+\sum_{i=1}^{k-1}\pi(\Pcal_{n/2,i})\,\pi(\Pcal_{n/2,k-i}),\label{eq:pi-decomp-n-k}
\end{align}
where the second, third, and fifth steps apply Lemmas~\ref{lem:ast-partition},
\ref{lem:ast-product}, and~\ref{lem:partition-translate}, respectively,
and the last step uses $\pi(\{\varnothing\})=1.$

The recurrence relations~(\ref{eq:pi-base-k})\textendash (\ref{eq:pi-decomp-n-k})
show that the hypothesis of Theorem~\ref{thm:N-n-k} is satisfied
for the function $N(n,k):=\pi(\Pcal_{n,k})$. As a result, Theorem~\ref{thm:N-n-k}
implies that
\[
\pi(\Pcal_{n,k})\leq\frac{(2+\sqrt{2})^{k-1}c^{k-1}}{\sqrt{k}}\left(\frac{n}{k}\right)^{k/2}
\]
for any $n\in\{1,2,4,8,16,\ldots\}$ and $k\geq1.$ This upper bound
in turn implies~(\ref{eq:pi-n-choose-k}) for any $n\geq1$ and $k\geq1$:
\begin{align*}
\pi(\Pcal_{n,k}) & \leq\pi(\Pcal_{2^{\lceil\log n\rceil},k})\\
 & \leq\frac{(2+\sqrt{2})^{k-1}c^{k-1}}{\sqrt{k}}\left(\frac{2^{\lceil\log n\rceil}}{k}\right)^{k/2}\\
 & \leq\frac{(2+\sqrt{2})^{k-1}c^{k-1}}{\sqrt{k}}\left(\frac{2n}{k}\right)^{k/2},
\end{align*}
where the first step uses Lemma~\ref{lem:partition-monotone}.
\end{proof}

\section{Fourier spectrum of decision trees}
\label{sec:fourier-main}

This section is devoted to the proof of our main result on the Fourier
spectrum of decision trees. Stated in its simplest terms, our result
shows that for any function $f\colon\pomon\to\{-1,0,1\}$ computable
by a decision tree of depth $d$, the sum of the absolute values of
the Fourier coefficients of order $k$ is at most
\[
C^{k}\sqrt{\binom{d}{k}(1+\ln n)^{k-1}},
\]
where $C\geq1$ is an absolute constant that does not depend on $n,d,k.$
Sections~\ref{subsec:Slicing-the-tree}\textendash \ref{subsec:Contiguous-intervals}
focus on partitioning the Fourier spectrum of $f$ into highly structured
parts and analyzing each in isolation. Sections~\ref{subsec:Generalization-to-elementary}
and~\ref{subsec:Main-result} then recombine these pieces using the
machinery of elementary families.

\subsection{\label{subsec:Slicing-the-tree}Slicing the tree}

Let $T$ be a given decision tree of depth $d$ in Boolean variables
$x_{1},x_{2},\ldots,x_{n}.$ For a set family $\Scal\subseteq\Pcal(\{1,2,\ldots,d\}),$
we define a real function $T|_{\Scal}\colon\pomo^{n}\to\Re$ by 
\begin{equation}
T|_{\Scal}(x)=\sum_{S\in\Scal}\sum_{v\in\pomo^{d}}\;T(v)\cdot2^{-d}\prod_{i\in S}v_{i}x_{T(v_{1}v_{2}\ldots v_{i-1})}.\label{eq:T_S}
\end{equation}
A straightforward but crucial observation is that $T|_{\Scal}$ is
additive with respect to $\Scal,$ in the following sense.
\begin{proposition}
\label{prop:T_S-additive}Let $T$ be a depth-$d$ decision tree.
Consider any set families $\Scal',\Scal''\subseteq\Pcal(\{1,2,\ldots,d\})$
with $\Scal'\cap\Scal''=\varnothing.$ Then
\[
T|_{\Scal'\cup\Scal''}=T|_{\Scal'}+T|_{\Scal''}.
\]
\end{proposition}

\begin{proof}
Immediate by taking $\Scal=\Scal'\cup\Scal''$ in the defining equation~(\ref{eq:T_S}).
\end{proof}
\noindent The relevance of~(\ref{eq:T_S}) to the Fourier spectrum
of decision trees is borne out by the following lemma.
\begin{lemma}
\label{lem:T_S-spectrum}Let $T$ be a decision tree of depth $d$
and degree at most $0,$ computing a function $f\colon\pomo^{n}\to\Re.$
Then
\begin{align*}
L_{k}f & =T|_{\Pcal_{d,k}}, &  & k=0,1,2,\ldots,n.
\end{align*}
\end{lemma}

\begin{proof}
By Proposition~\ref{prop:tree-computation},
\begin{align}
f(x) & =\sum_{v\in\pomo^{d}}\;T(v)\cdot\prod_{i=1}^{d}\frac{1+v_{i}x_{T(v_{1}v_{2}\ldots v_{i-1})}}{2}\nonumber \\
 & =\sum_{v\in\pomo^{d}}\;T(v)\cdot2^{-d}\sum_{S\subseteq\{1,2,\ldots,d\}}\prod_{i\in S}v_{i}x_{T(v_{1}v_{2}\ldots v_{i-1})}\nonumber \\
 & =\sum_{k=0}^{d}\,\sum_{S\in\Pcal_{d,k}}\,\sum_{v\in\pomo^{d}}\;T(v)\cdot2^{-d}\prod_{i\in S}v_{i}x_{T(v_{1}v_{2}\ldots v_{i-1})}.\label{eq:f-T-expanded}
\end{align}
Since $\deg(T)\leq0,$ the coefficients $T(v)$ for $v\in\pomo^{d}$
are real numbers. Moreover, for any $v\in\pomo^{d}$ and $S\subseteq\{1,2,\ldots,d\},$
the definition of a decision tree ensures that the product $\prod_{i\in S}v_{i}x_{T(v_{1}v_{2}\ldots v_{i-1})}$
is a signed monomial of degree $|S|.$ We conclude from~(\ref{eq:f-T-expanded})
that the degree-$k$ homogeneous part of $f$ is 
\begin{align*}
L_{k}f & =\sum_{S\in\Pcal_{d,k}}\,\sum_{v\in\pomo^{d}}\;T(v)\cdot2^{-d}\prod_{i\in S}v_{i}x_{T(v_{1}v_{2}\ldots v_{i-1})}\\
 & =T|_{\Pcal_{d,k}.}
\end{align*}
In particular, $L_{k}f=0$ for $k\geq d+1.$
\end{proof}
Looking ahead, much of our analysis of the Fourier spectrum of decision
trees $T$ focuses on $T|_{\Ecal}$ for elementary families $\Ecal\subseteq\Pcal_{d,k}.$
This analysis proceeds by induction, with the following lemma required
as part of the inductive step. The reader may wish to review
Sections~\ref{sec:fourier} and~\ref{sec:gen-d-t} for the meaning of the symbols
$\Tcal,\Tcal^*,$ and $\norm{\cdot}.$
\begin{lemma}
\label{lem:collect-like-leaves}Let $T\in\Tcal(n,d,p,k)$ be a 
decision tree and $\Scal\subseteq\Pcal(\{1,2,\ldots,d\}).$ Define
$m=\max_{v\in\pomo^{d}}\norm{T(v)}.$ Then for each $i=1,2,\ldots,\binom{n}{k},$
there is a real $0\leq p_{i}\leq1$ and a decision tree $U_{i}\in\Tcal^{*}(n,d,p_{i},0)$
such that
\begin{align*}
 & p=\sum_{i=1}^{\binom{n}{k}}p_{i},\\
 & \norm{T|_{\Scal}}\leq m\sum_{i=1}^{\binom{n}{k}}\norm{U_{i}|_{\Scal}}.
\end{align*}
\end{lemma}

\begin{proof}
Let $\phi=\sum_{S\subseteq\oneton}\hat{\phi}(S)\chi_{S}$ be an arbitrary
nonzero polynomial with $\norm\phi\leq1.$ Consider the random variable
$X\in\{\pm\chi_{S}:\hat{\phi}(S)\ne0\}$ distributed according to
\[
\Prob[X=\sigma\chi_{S}]=\frac{|\hat{\phi}(S)|}{\norm\phi}\left(\frac{1}{2}+\frac{\norm\phi}{2}\cdot\sigma\sign\hat{\phi}(S)\right)
\]
for all $\sigma\in\pomo$ and $S\subseteq\oneton$. Then
\begin{align*}
\Exp X & =\sum_{S\subseteq\oneton}\,\sum_{\sigma\in\pomo}\sigma\chi_{S}\cdot\frac{|\hat{\phi}(S)|}{\norm\phi}\left(\frac{1}{2}+\frac{\norm\phi}{2}\cdot\sigma\sign\hat{\phi}(S)\right)\\
 & =\sum_{S\subseteq\oneton}\chi_{S}\cdot\frac{|\hat{\phi}(S)|}{\norm\phi}\cdot\norm\phi\cdot\sign\hat{\phi}(S)\\
 & =\phi(x).
\end{align*}
In conclusion, $\phi$ can be viewed as the \emph{expected value} of
a random variable $X\in\{\pm\chi_{S}:\hat{\phi}(S)\ne0\}.$

We may assume that $T$ has at least one nonzero leaf, since otherwise
the lemma holds trivially with $p_{1}=p_{2}=\cdots=p_{\binom{n}{k}}=p=0.$
Recall from the definition of $m$ that
$\norm{T(v)/m}=\norm{T(v)}/m\leq1$ for each $v.$
Now the previous paragraph implies that for every leaf $v\in\pomo^{d}$
with $T(v)\ne0,$ the polynomial $T(v)/m$ is the expected value of
a random variable $X_{v}$ whose support is contained in the set of
the nonzero degree-$k$ monomials of $T(v)$ with $\pm1$ coefficients.
The joint distribution of the $X_{v}$ is immaterial for our purposes,
but for concreteness let us declare them to be independent. Then
\begin{align*}
T|_{\Scal}(x) & =m\sum_{S\in\Scal}\sum_{v\in\pomo^{d}}\;\frac{T(v)}{m}\cdot2^{-d}\prod_{i\in S}v_{i}x_{T(v_{1}v_{2}\ldots v_{i-1})}\\
 & =m\sum_{S\in\Scal}\sum_{\substack{v\in\pomo^{d}:\\
T(v)\ne0
}
}\;\Exp[X_{v}]\cdot2^{-d}\prod_{i\in S}v_{i}x_{T(v_{1}v_{2}\ldots v_{i-1})}\\
 & =m\Exp\left[\sum_{S\in\Scal}\sum_{\substack{v\in\pomo^{d}:\\
T(v)\ne0
}
}X_{v}\cdot2^{-d}\prod_{i\in S}v_{i}x_{T(v_{1}v_{2}\ldots v_{i-1})}\right].
\end{align*}
Applying Proposition~\ref{prop:fourier-norm-convex},
\begin{align}
\norm{T|_{\Scal}} & \leq m\Exp\NORM{\sum_{S\in\Scal}\sum_{\substack{v\in\pomo^{d}:\\
T(v)\ne0
}
}X_{v}\cdot2^{-d}\prod_{i\in S}v_{i}x_{T(v_{1}v_{2}\ldots v_{i-1})}}.\label{eq:T-S-extreme-points}
\end{align}
In the last expression, each random variable $X_{v}$ is a signed
monomial of degree $k$ that does not contain any of the variables
$x_{T(\varepsilon)},x_{T(v_{1})},\ldots,x_{T(v_{1}v_{2}\ldots v_{d-1})}$
queried along the path from the root to $v.$ Therefore, the expectation
in~(\ref{eq:T-S-extreme-points}) is over $\norm{U|_{\Scal}}$ for
some trees $U\in\Tcal^{*}(n,d,p,k)$. We conclude that there is a
fixed decision tree $U\in\Tcal^{*}(n,d,p,k)$ with 
\begin{equation}
\norm{T|_{\Scal}}\leq m\,\norm{U|_{\Scal}}.\label{eq:T-U-m}
\end{equation}

Finally, decompose
\[
U|_{\Scal}=\sum_{S\in\Pcal_{n,k}}U_{S}|_{\Scal}\cdot\chi_{S},
\]
where $U_{S}$ is the depth-$d$ decision tree given by
\[
U_{S}(v)=\begin{cases}
U(v) & \text{if }|v|\leq d-1,\\
-1 & \text{if }|v|=d\text{ and }U(v)=-\chi_{S},\\
1 & \text{if }|v|=d\text{ and }U(v)=\chi_{S},\\
0 & \text{otherwise.}
\end{cases}
\]
In other words, $U_{S}$ is the decision tree obtained from $U$ by
setting to $1$ every leaf labeled $\chi_{S}$, setting to $-1$ every
leaf labeled $-\chi_{S}$, and setting all other leaves to $0$. It
is clear that the densities of the $U_{S}$ sum to the density of
$U$. We conclude that $U_{S}\in\Tcal^{*}(n,d,p_{S},0)$ for some
reals $0\leq p_{S}\leq1$ with $\sum_{S\in\Pcal_{n,k}}p_{S}=p.$ Moreover,
\begin{align*}
\norm{T|_{\Scal}} & \leq m\,\norm{U|_{\Scal}}\\
 & \leq m\sum_{S\in\Pcal_{n,k}}\norm{U_{S}|_{\Scal}\cdot\chi_{S}}\\
 & \leq m\sum_{S\in\Pcal_{n,k}}\norm{U_{S}|_{\Scal}},
\end{align*}
where the first step is a restatement of~(\ref{eq:T-U-m}); the second
step applies Proposition~\ref{prop:fourier-norm-convex}; and the
last step is justified by Proposition~\ref{prop:fourier-norm-submult}.
In summary, the decision trees $U_{1},U_{2},\ldots,U_{\binom{n}{k}}$
in the statement of the lemma can be taken to be the $U_{S}$, in
arbitrary order.
\end{proof}

\subsection{Analytic preliminaries}
\label{subsec:analytic-preliminaries}

For positive integers $m$ and $k$, define
\[
\Lambda_{m,k}(p)=\begin{cases}
0 & \text{if }p=0,\\
{\displaystyle p\,\sqrt{\left(\frac{1}{k}\ln\frac{e^{k}m^{k-1}}{p}\right)^{k}}} & \text{if }0<p\leq1/m,\\
\rule{0mm}{9mm}{\displaystyle p\,\sqrt{\left(\ln\frac{e}{p}\right)(\ln em)^{k-1}}} & \text{if }1/m<p\leq1.
\end{cases}
\]
Our bound for the Fourier spectrum of decision trees is in terms of
this function. As preparation for our main result, we now collect
the analytic properties of $\Lambda_{m,k}$ that we will need.
\begin{lemma}
\label{lem:Lambda-cont-monotone-concave}Let $m$ and $k$ be any
positive integers. Then:
\begin{enumerate}
\item $\Lambda_{m,k}$ is continuous on $[0,1];$
\item $\Lambda_{m,k}$ is monotonically increasing on $[0,1];$
\item $\Lambda_{m,k}$ is concave on $[0,1].$
\end{enumerate}
\end{lemma}

\begin{proof}
(i)~The continuity on $(0,1/m)\cup(1/m,1]$ is immediate. The continuity
at $p=0$ and $p=1/m$ follows by examining the one-sided limits at
those points, which are $0$ and $(\ln em)^{k/2}/m,$ respectively.

(ii)~Considering the derivative $\Lambda_{m,k}'$ separately on $(0,1/m)$
and $(1/m,1]$, one finds in both cases that the derivative is positive:
\[
\Lambda_{m,k}'(p)=\begin{cases}
{\displaystyle \sqrt{\left(\frac{1}{k}\ln\frac{e^{k}m^{k-1}}{p}\right)^{k}}\left(1-\frac{k}{2\ln(e^{k}m^{k-1}/p)}\right)} & \text{if }0<p<1/m,\\
{\displaystyle \rule{0mm}{9mm}\left(\sqrt{\ln\frac{e}{p}}-\frac{1}{2\sqrt{\ln(e/p)}}\right)\sqrt{(\ln em)^{k-1}}} & \text{if }1/m<p\leq1.
\end{cases}
\]
Since $\Lambda_{m,k}$ is continuous on $[0,1]$, it follows that
$\Lambda_{m,k}$ is monotonically increasing on $[0,1].$

(iii)~The one-sided derivatives of $\Lambda_{m,k}$ at $p=1/m$ are
both $(\ln em)^{\frac{k-2}{2}}\ln(\sqrt{e}m)$. Along with the 
formulas derived in (ii) for $\Lambda_{m,k}'$ on $(0,1/m)$ and $(1/m,1],$
this shows that $\Lambda_{m,k}$ is continuously differentiable
on $(0,1]$. The formulas in~(ii) further reveal that $\Lambda_{m,k}'$
is monotonically decreasing on $(0,1/m)$ and on
$(1/m,1].$ Indeed, the formula for $(0,1/m)$ shows that
$\Lambda'_{m,k}$ is the product of two nonnegative
factors, each of which clearly decreases with $p$; the
formula for $(1/m,1]$ shows that $\Lambda'_{m,k}$ is a
constant multiple of
$\sqrt{\ln(e/p)}-1/(2\sqrt{\ln(e/p)}),$ where the
minuend decreases with $p$ and the subtrahend increases
with $p.$ 

Since $\Lambda'_{m,k}$ is monotonically 
decreasing on $(0,1/m)$ and on $(1/m,1],$ and
continuous on $(0,1],$
we conclude that $\Lambda'_{m,k}$
is monotonically decreasing on $(0,1],$ which in turn makes $\Lambda_{m,k}$
concave on $(0,1].$ Since $\Lambda_{m,k}$ is continuous at $0,$
we conclude that $\Lambda_{m,k}$ is concave on the entire interval
$[0,1].$~
\end{proof} 
The function $\Lambda_{m,k}$ arises in our work not as the
closed form defined above, but rather as a certain optimization problem from
an inductive argument. We now describe this optimization view and prove its
equivalence with the above definition.
\begin{lemma}
\label{lem:Lambda-opt}Let $m$ and $k$ be positive integers. Then
for $0<p\leq1,$
\begin{equation}
\Lambda_{m,k}(p)=p\max\left\{ \prod_{i=1}^{k}\sqrt{\ln ex_{i}}:x_{i}\geq1\text{ and }x_{1}x_{2}\ldots x_{i}\leq\frac{m^{i-1}}{p}\text{ for all }i\right\} .\label{eq:Lambda-opt}
\end{equation}
\end{lemma}

\begin{proof}
For $k=1,$ the left-hand side and right-hand side are clearly $p\sqrt{\ln(e/p)}.$
In what follows, we treat the complementary case $k\geq2.$

For $0<p\leq1/m,$ the upper bound in~(\ref{eq:Lambda-opt}) follows
by taking $x_{1}=x_{2}=\cdots=x_{k}=(m^{k-1}/p)^{1/k}$. For $1/m<p\leq1,$
the upper bound follows by setting $x_{1}=1/p$ and $x_{2}=\cdots=x_{k}=m.$

For the lower bound in~(\ref{eq:Lambda-opt}), fix reals $x_{1},x_{2},\ldots,x_{k}\geq1$
with $x_{1}\leq1/p$ and $x_{1}x_{2}\ldots x_{k}\leq m^{k-1}/p$.
Then
\begin{align}
\sqrt{\ln ex_{1}}\cdot\prod_{i=2}^{k}\sqrt{\ln ex_{i}} & \leq\sqrt{\ln ex_{1}}\left(\frac{1}{k-1}\ln e^{k-1}x_{2}\ldots x_{k}\right)^{(k-1)/2}\nonumber \\
 & \leq\sqrt{\ln ex_{1}}\left(\frac{1}{k-1}\ln\frac{e^{k-1}m^{k-1}}{px_{1}}\right)^{(k-1)/2},\label{eq:opt-rewritten}
\end{align}
where the first step applies the AM\textendash GM inequality. Elementary
calculus shows that~(\ref{eq:opt-rewritten}) as a function of $x_{1}$
is monotonically increasing on $[1,(m^{k-1}/p)^{1/k}]$ and monotonically
decreasing on $[(m^{k-1}/p)^{1/k},m^{k-1}/p]$. Recalling that $1\leq x_{1}\leq1/p,$
we conclude that~(\ref{eq:opt-rewritten}) is maximized at 
\begin{align*}
x_{1} & =\min\left(\left(\frac{m^{k-1}}{p}\right)^{1/k},\frac{1}{p}\right)\\
 & =\begin{cases}
(m^{k-1}/p)^{1/k} & \text{if }0<p\leq1/m,\\
1/p & \text{if }1/m<p\leq1.
\end{cases}
\end{align*}
Making this substitution shows that (\ref{eq:opt-rewritten}) does
not exceed $\Lambda_{m,k}(p).$
\end{proof}
\noindent This optimization view of $\Lambda_{m,k}$ implies a host
of useful facts that would be a hassle to prove directly. We state
them as corollaries below.
\begin{corollary}
\label{cor:Lambda-divided-p-monotone}Let $m$ and $k$ be positive
integers. Then for all $p,q\in[0,1],$
\[
q\Lambda_{m,k}(p)\leq\Lambda_{m,k}(pq).
\]
\end{corollary}

\begin{proof}
If $p=0$ or $q=0,$ the left-hand side and right-hand side both vanish.
If $p,q\in(0,1],$ the claim can be equivalently stated as $\Lambda_{m,k}(p)/p\leq\Lambda_{m,k}(pq)/pq,$
which in turn amounts to saying that $\Lambda_{m,k}(p)/p$ is monotonically
nonincreasing in $p\in(0,1].$ This monotonicity is immediate from
Lemma~\ref{lem:Lambda-opt}.
\end{proof}
\begin{corollary}
\label{cor:Lambda-collapsing}Let $m,k,\ell$ be positive integers.
Then for all $p,q\in[0,1],$
\[
\Lambda_{m,k}(p)\,\Lambda_{m,\ell}\!\left(\frac{q}{m}\right)\leq\frac{\Lambda_{m,k+\ell}(pq)}{m}.
\]
\end{corollary}

\begin{proof}
If $p=0$ or $q=0,$ the left-hand side and right-hand side both vanish.
In what follows, we treat $p,q\in(0,1].$ By Lemma~\ref{lem:Lambda-opt},
\begin{equation}
\Lambda_{m,k}(p)\,\Lambda_{m,\ell}\!\left(\frac{q}{m}\right)=\frac{pq}{m}\max\left\{ \prod_{i=1}^{k+\ell}\sqrt{\ln ex_{i}}\right\} ,\label{eq:lambda-product}
\end{equation}
where the maximum is over all $x_{1},x_{2},\ldots,x_{k+\ell}\geq1$
such that
\begin{align}
 & x_{1}x_{2}\ldots x_{i}\leq\frac{m^{i-1}}{p}, &  & i=1,2,\ldots,k,\label{eq:x-low}\\
 & x_{k+1}x_{k+2}\ldots x_{i}\leq\frac{m^{i-k-1}}{q/m}, &  & i=k+1,\ldots,k+\ell.\label{eq:x-high}
\end{align}
Equations~(\ref{eq:x-low}) and~(\ref{eq:x-high}) imply that the
maximum in~(\ref{eq:lambda-product}) is over $x_{1},x_{2},\ldots,x_{k+\ell}\geq1$
that satisfy, among other things, $x_{1}x_{2}\ldots x_{i}\leq m^{i-1}/(pq)$
for $i=1,2,\ldots,k+\ell.$ Now Lemma~\ref{lem:Lambda-opt} implies
that the right-hand side of~(\ref{eq:lambda-product}) is at most~$\Lambda_{m,k+\ell}(pq)/m.$
\end{proof}
\begin{corollary}
\label{lem:Lambda-sqrt}Let $m$ and $k$ be positive integers. Then
for all $p\in[0,1],$
\begin{equation}
\Lambda_{m,k}(p)\leq\sqrt{2^{k}p}\cdot\Lambda_{m,k}(\sqrt{p}).\label{eq:Lambda-opt-1}
\end{equation}
\end{corollary}

\begin{proof}
For $p=0$, the left-hand side and right-hand side both vanish. For
$p\in(0,1],$ we have:
\begin{align*}
\Lambda_{m,k}(p) & =p\max\left\{ \prod_{i=1}^{k}\sqrt{\ln ex_{i}}:x_{i}\geq1\text{ and }x_{1}x_{2}\ldots x_{i}\leq\frac{m^{i-1}}{p}\text{ for all }i\right\} \\
 & \leq p\max\left\{ \prod_{i=1}^{k}\sqrt{\ln ex_{i}^{2}}:x_{i}\geq1\text{ and }x_{1}x_{2}\ldots x_{i}\leq\frac{m^{i-1}}{\sqrt{p}}\text{ for all }i\right\} \\
 & \leq\sqrt{2^{k}}\,p\max\left\{ \prod_{i=1}^{k}\sqrt{\ln ex_{i}}:x_{i}\geq1\text{ and }x_{1}x_{2}\ldots x_{i}\leq\frac{m^{i-1}}{\sqrt{p}}\text{ for all }i\right\} \\
 & =\sqrt{2^{k}p}\cdot\Lambda_{m,k}(\sqrt{p}),
\end{align*}
where the first and last steps use Lemma~\ref{lem:Lambda-opt}.
\end{proof}

\subsection{\label{subsec:Contiguous-intervals}Contiguous intervals}

We have reached a focal point of this paper, where we analyze $T|_{\Ecal}$
for arbitrary decision trees $T$ and ``canonical'' elementary families
$\Ecal$. The families that we allow are those of the form
\[
\Ecal=\binom{I_{1}}{k_{1}}*\binom{I_{2}}{k_{2}}*\cdots*\binom{I_{\ell}}{k_{\ell}},
\]
where $k_{1},k_{2},\ldots,k_{\ell}\in\{1,2\}$ and the integer intervals
$I_{1},I_{2},\ldots,I_{\ell}$ form a partition of $\{1,2,\ldots,d\}$
with $d$ being the depth of $T.$ The proof proceeds by induction
on $\ell,$ with Lemmas~\ref{lem:T_S-spectrum}, \ref{lem:collect-like-leaves},
and the analytic properties of $\Lambda_{m,k}$ applied in the inductive
step. We will later generalize this result to arbitrary elementary
families $\Ecal$ and, from there, to all of $\Pcal_{d,k}$ via the
results of Section~\ref{sec:Elementary-set-families}.
\begin{theorem}
\label{thm:norm-canonical}Let $T\in\Tcal^{*}(n,d,p,0)$ be given,
for some $0\leq p\leq1$ and integers $n,d\geq1$. Let $\ell\geq1.$
Let $I_{1},I_{2},\ldots,I_{\ell}$ be pairwise disjoint integer intervals
with $I_{1}\cup I_{2}\cup\cdots\cup I_{\ell}=\{1,2,\ldots,d\},$ and
let $k_{1},k_{2},\ldots,k_{\ell}\in\{1,2\}.$ Abbreviate $k=k_{1}+k_{2}+\cdots+k_{\ell}.$
Then
\begin{equation}
\NORM{T|_{\binom{I_{1}}{k_{1}}*\binom{I_{2}}{k_{2}}*\cdots*\binom{I_{\ell}}{k_{\ell}}}}\leq2C^{k}\,12^{\ell-1}\Lambda_{n^{2},k}(p)\prod_{i=1}^{\ell}\binom{|I_{i}|}{k_{i}}^{1/2},\label{eq:norm-canonical}
\end{equation}
where $C\geq1$ is the absolute constant from Theorem~\emph{\ref{thm:tal-depth-2}}.
\end{theorem}

\begin{proof}
The proof is by induction on $\ell.$ The base case $\ell=1$ corresponds
to $I_{1}=\{1,2,\ldots,d\}$. Let $f\colon\pomon\to\{-1,0,1\}$ be
the function computed by $T.$ If $f\equiv0,$ we have $T|_{\binom{I_{1}}{k_{1}}}\equiv0$
and the bound holds trivially. In the complementary case $f\not\equiv0,$
recall from Fact~\ref{fact:dns-prob-nonzero} that
\begin{equation}
\Prob_{x\in\pomon}[f(x)\ne0]=p.\label{eq:prob-f-nonzero}
\end{equation}
Then
\begin{align*}
\norm{T|_{\binom{I_{1}}{k_{1}}}} & =\norm{L_{k_{1}}f}\\
 & \leq\binom{|I_{1}|}{k_{1}}^{1/2}C^{k_{1}}p\prod_{i=1}^{k_{1}}\sqrt{\ln\frac{en^{i-1}}{p}}\\
 & \leq\binom{|I_{1}|}{k_{1}}^{1/2}\cdot2C^{k_{1}}p\prod_{i=1}^{k_{1}}\sqrt{\ln\frac{en^{i-1}}{\sqrt{p}}}\\
 & \leq\binom{|I_{1}|}{k_{1}}^{1/2}\cdot2C^{k_{1}}\Lambda_{n^{2},k_{1}}(p)\\
 & =\binom{|I_{1}|}{k_{1}}^{1/2}\cdot2C^{k}\Lambda_{n^{2},k}(p),
\end{align*}
where the first step is valid by Lemma~\ref{lem:T_S-spectrum}; the
second step uses Theorem~\ref{thm:tal-depth-2} along with~(\ref{eq:prob-f-nonzero})
and~$k_{1}\leq2$; and the fourth step applies Lemma~\ref{lem:Lambda-opt}
with $m=n^2$ and $k=k_1\leq2.$
This settles the base case. 
We note that the last derivation could be sharpened so as to replace
$\Lambda_{n^2,k}$ with $\Lambda_{n,k}$; however, this savings would not make a
difference because the bound in the inductive step requires $\Lambda_{n^2,k}.$

We now turn to the inductive step, $\ell\geq2.$ If $k_{j}>|I_{j}|$
for some $j,$ then 
\[
T|_{\binom{I_{1}}{k_{1}}*\binom{I_{2}}{k_{2}}*\cdots*\binom{I_{\ell}}{k_{\ell}}}=T|_{\varnothing}=0,
\]
and the claimed bound holds trivially. We may therefore assume that
$k_{j}\leq|I_{j}|$ for every $j=1,2,\ldots,\ell.$ This means in
particular that the intervals $I_{1},I_{2},\ldots,I_{\ell}$ are nonempty.
Furthermore, by renumbering the intervals if necessary, we may assume
that $I_{1}<I_{2}<\cdots<I_{\ell}.$ Put $d'=\max I_{\ell-1},$ so
that $I_{\ell}=\{d'+1,d'+2,\ldots,d\}.$ Abbreviate
\begin{align*}
\Scal' & =\binom{I_{1}}{k_{1}}*\binom{I_{2}}{k_{2}}*\cdots*\binom{I_{\ell-1}}{k_{\ell-1}},\\
\Scal & =\Scal'*\binom{I_{\ell}}{k_{\ell}}.
\end{align*}
For $j=0,1,2,\ldots,$ define a depth-$d'$ decision tree $T_{j}'$
by
\[
T'_{j}(v)=\begin{cases}
T(v) & \text{if }v\in\pomo^{\leq d'-1},\\
T_{v}|_{\binom{\{1,2,\ldots,|I_{\ell}|\}}{k_{\ell}}} & \text{if }v\in\pomo^{d'}\text{ and }\dns(T_{v})\in(3^{-j-1},3^{-j}]\\
0 & \text{otherwise.}
\end{cases},
\]
This definition corresponds to part~(i) of the program set forth in the introduction (page~\pageref{three-part-program}).
Observe that $T_{j}'$ is a valid decision tree in that for every
leaf $v\in\pomo^{d'}$, the label $T_{j}'(v)\in\Re[x_{1},x_{2},\ldots,x_{n}]$
is a function that does not depend on any of the variables 
\begin{equation}
x_{T(\varepsilon)},x_{T(v_{1})},x_{T(v_{1}v_{2})},\ldots,x_{T(v_{1}v_{2}\ldots v_{d'-1})}\label{eq:vars-queried}
\end{equation}
queried along the path from the root to $v.$ Indeed, recall from
Lemma~\ref{lem:T_S-spectrum} that $T_{v}|_{\binom{\{1,2,\dots,|I_{\ell}|\}}{k_{\ell}}}$
is the $k_{\ell}$-th homogeneous part of the function computed by
the subtree $T_{v}$, which by definition does not use any of the
variables~(\ref{eq:vars-queried}). We also note that all but finitely
many of the trees $T_{0},T_{1},T_{2},\ldots$ are identically zero;
however, working with the infinite sequence is more convenient from
the point of view of notation and calculations.

The weighted densities of $T_{0}',T_{1}',T_{2}',\ldots$ are given
by
\begin{align}
\sum_{j=0}^{\infty}3^{-j}\dns(T_{j}') & =\sum_{j=0}^{\infty}3^{-j}\Prob_{v\in\pomo^{d'}}[T_{j}'(v)\ne0]\nonumber \\
 & \leq\sum_{j=0}^{\infty}3^{-j}\Prob_{v\in\pomo^{d'}}[3^{-j-1}<\dns(T_{v})\leq3^{-j}]\nonumber \\
 & \leq3\Exp_{v\in\pomo^{d'}}\dns(T_{v})\nonumber \\
 & =3\dns(T)\nonumber \\
 & =3p.\label{eq:weighted-densities}
\end{align}
The relevance of $T'_{j}$ to our analysis of $T|_{\Scal}$ is clear
from the following claims, whose proofs we will present shortly.
\begin{claim}
\label{claim:T_S-T_j}$T|_{\Scal}=\sum_{j=0}^{\infty}T_{j}'|_{\Scal'}.$
\end{claim}

\begin{claim}
\label{claim:T-j-norm}For $j=0,1,2,\ldots,$ one has
\[
\norm{T_{j}'|_{\Scal'}}\leq8C^{k}\,12^{\ell-2}\binom{|I_{1}|}{k_{1}}^{1/2}\cdots\binom{|I_{\ell}|}{k_{\ell}}^{1/2}\cdot\sqrt{3^{-j}}\Lambda_{n^{2},k}(\sqrt{3^{-j}}\dns(T_{j}')).
\]
\end{claim}

We now complete the proof of the theorem. Set $s=\sum_{i=0}^{\infty}\sqrt{3^{-i}}=2.3660\ldots.$
Then
\begin{align}
\sum_{j=0}^{\infty}\sqrt{3^{-j}}\Lambda_{n^{2},k}(\sqrt{3^{-j}}\dns(T_{j}')) & =s\sum_{j=0}^{\infty}\frac{\sqrt{3^{-j}}}{s}\Lambda_{n^{2},k}(\sqrt{3^{-j}}\dns(T_{j}'))\nonumber \\
 & \leq s\Lambda_{n^{2},k}\left(\sum_{j=0}^{\infty}\frac{\sqrt{3^{-j}}}{s}\cdot\sqrt{3^{-j}}\dns(T_{j}')\right)\nonumber \\
 & \leq3\Lambda_{n^{2},k}\left(\frac{s}{3}\sum_{j=0}^{\infty}\frac{\sqrt{3^{-j}}}{s}\cdot\sqrt{3^{-j}}\dns(T_{j}')\right)\nonumber \\
 & \leq3\Lambda_{n^{2},k}(p),\label{eq:sum-of-weighted-IH}
\end{align}
where the second step is valid by Lemma~\ref{lem:Lambda-cont-monotone-concave}~(iii);
the third step uses Corollary~\ref{cor:Lambda-divided-p-monotone}
with $q=s/3$; and the final step is justified by~(\ref{eq:weighted-densities})
and Lemma~\ref{lem:Lambda-cont-monotone-concave}~(ii). As a result,
\begin{align*}
\norm{T|_{\Scal}} & \leq\sum_{j=0}^{\infty}\norm{T_{j}'|_{\Scal'}}\\
 & \leq8C^{k}\,12^{\ell-2}\binom{|I_{1}|}{k_{1}}^{1/2}\cdots\binom{|I_{\ell}|}{k_{\ell}}^{1/2}\sum_{j=0}^{\infty}\sqrt{3^{-j}}\Lambda_{n^{2},k}(\sqrt{3^{-j}}\dns(T_{j}'))\\
 & \leq2C^{k}\,12^{\ell-1}\binom{|I_{1}|}{k_{1}}^{1/2}\cdots\binom{|I_{\ell}|}{k_{\ell}}^{1/2}\Lambda_{n^{2},k}(p),
\end{align*}
where the first step is valid by Proposition~\ref{prop:fourier-norm-convex}
and Claim~\ref{claim:T_S-T_j}, bearing in mind once again that all
but finitely many of the $T_{j}'|_{\Scal'}$ are identically zero;
the second step is a substitution from Claim~\ref{claim:T-j-norm};
and the final step uses~(\ref{eq:sum-of-weighted-IH}). This completes
the inductive step.
\end{proof}
\begin{proof}[Proof of Claim~\emph{\ref{claim:T_S-T_j}}]
Let $T'$ be the depth-$d'$ decision tree given by
\[
T'(v)=\begin{cases}
T(v) & \text{if }v\in\pomo^{\leq d'-1},\\
T_{v}|_{\binom{\{1,2,\ldots,|I_{\ell}|\}}{k_{\ell}}} & \text{if }v\in\pomo^{d'}.
\end{cases}
\]
This definition implies that 
\begin{align*}
T'(v)=\begin{cases}
T'_{0}(v)=T'_{1}(v)=T'_{2}(v)=\cdots & \text{if }v\in\pomo^{\leq d'-1},\\
T'_{0}(v)+T'_{1}(v)+T'_{2}(v)+\cdots & \text{if }v\in\pomo^{d'}.
\end{cases}
\end{align*}
As a result,
\begin{align}
T'|_{\Scal'} & =\sum_{S\in\Scal'}\sum_{v\in\pomo^{d'}}\;\left(\sum_{j=0}^{\infty}T'_{j}(v)\right)\cdot2^{-d'}\prod_{i\in S}v_{i}x_{T'(v_{1}v_{2}\ldots v_{i-1})}\nonumber \\
 & =\sum_{j=0}^{\infty}\sum_{S\in\Scal'}\sum_{v\in\pomo^{d'}}\;T'_{j}(v)\cdot2^{-d'}\prod_{i\in S}v_{i}x_{T'_{j}(v_{1}v_{2}\ldots v_{i-1})}\nonumber \\
 & =\sum_{j=0}^{\infty}T'_{j}|_{\Scal'}.\label{eq:T-prime-Scal-prime-decomposition}
\end{align}
Thus, the proof will be complete once we show that $T'|_{\Scal'}=T|_{\Scal}.$

Since $\Scal$ is the family of sets $S$ expressible as $S=S'\cup S''$
with $S'\in\Scal'$ and $S''\in\binom{I_{\ell}}{k_{\ell}},$ we have
\begin{align}
T|_{\Scal} & =\sum_{S\in\Scal}\sum_{v\in\pomo^{d}}\;T(v)\cdot2^{-d}\prod_{i\in S}v_{i}x_{T(v_{1}v_{2}\ldots v_{i-1})}\nonumber \\
 & =\sum_{S'\in\Scal'}\sum_{S''\in\binom{I_{\ell}}{k_{\ell}}}\sum_{v\in\pomo^{d}}\;T(v)\cdot2^{-d}\prod_{i\in S'\cup S''}v_{i}x_{T(v_{1}v_{2}\ldots v_{i-1})}.\label{eq:T-S-intermediate}
\end{align}
Recall that $\Scal'\subseteq\Pcal(\{1,2,\ldots,d'\})$ and $I_{\ell}=\{d'+1,d'+2,\ldots,d\}.$
As a result,~(\ref{eq:T-S-intermediate}) yields
\begin{multline*}
T|_{\Scal}=\sum_{S'\in\Scal'}\sum_{S''\in\binom{I_{\ell}}{k_{\ell}}}\sum_{\substack{v'\in\pomo^{d'}\\
v''\in\pomo^{d-d'}
}
}\;T(v'v'')\cdot2^{-d}\prod_{i\in S'}v'_{i}x_{T(v'_{1}v'_{2}\ldots v'_{i-1})}\\
\times\prod_{i\in S''}v''_{i-d'}x_{T(v'v_{1}''v_{2}''\ldots v''_{i-1-d'})}.
\end{multline*}
 A change of index now gives
\begin{multline*}
T|_{\Scal}=\sum_{S'\in\Scal'}\sum_{S''\in\binom{\{1,2,\ldots,|I_{\ell}|\}}{k_{\ell}}}\sum_{\substack{v'\in\pomo^{d'}\\
v''\in\pomo^{d-d'}
}
}\!T(v'v'')\cdot2^{-d}\prod_{i\in S'}v'_{i}x_{T(v'_{1}v'_{2}\ldots v'_{i-1})}\\
\times\prod_{i\in S''}v''_{i}x_{T(v'v_{1}''v_{2}''\ldots v''_{i-1})}.
\end{multline*}
Since $T(v'v'')=T_{v'}(v'')$ and $T(v'v_{1}''v_{2}''\ldots v''_{i-1})=T_{v'}(v''_{1}v_{2}''\ldots v''_{i-1}),$
we arrive at
\begin{multline*}
T|_{\Scal}=\sum_{S'\in\Scal'}\sum_{v'\in\pomo^{d'}}2^{-d'}\prod_{i\in S'}v'_{i}x_{T(v'_{1}v'_{2}\ldots v'_{i-1})}\\
\quad\times\left(\sum_{S''\in\binom{\{1,2,\ldots,|I_{\ell}|\}}{k_{\ell}}}\sum_{v''\in\pomo^{d-d'}}\!T_{v'}(v'')\cdot2^{-d+d'}\prod_{i\in S''}v''_{i}x_{T_{v'}(v_{1}''v_{2}''\ldots v''_{i-1})}\right).
\end{multline*}
The large parenthesized expression is by definition $T_{v'}|_{\binom{\{1,2,\ldots,|I_{\ell}|\}}{k_{\ell}}}=T'(v'),$
whence
\begin{align}
T|_{\Scal} & =\sum_{S'\in\Scal'}\sum_{v'\in\pomo^{d'}}T'(v')\cdot2^{-d'}\prod_{i\in S'}v'_{i}x_{T(v'_{1}v'_{2}\ldots v'_{i-1})}\nonumber \\
 & =\sum_{S'\in\Scal'}\sum_{v'\in\pomo^{d'}}T'(v')\cdot2^{-d'}\prod_{i\in S'}v'_{i}x_{T'(v'_{1}v'_{2}\ldots v'_{i-1})}\nonumber \\
 & =T'|_{\Scal'}.\label{eq:T-T-prime-Scal-prime}
\end{align}
By~(\ref{eq:T-prime-Scal-prime-decomposition}) and~(\ref{eq:T-T-prime-Scal-prime}),
the proof is complete.
\end{proof}
\begin{proof}[Proof of Claim~\emph{\ref{claim:T-j-norm}}]
Recall from Lemma~\ref{lem:T_S-spectrum} that $T_{v}|_{\binom{\{1,2,\dots,|I_{\ell}|\}}{k_{\ell}}}$
is the $k_{\ell}$-th homogeneous part of the function computed by
the subtree $T_{v}$ of $T.$ This implies that $T_{j}'\in\Tcal(n,d',\dns(T_{j}'),k_{\ell})$.
Moreover, every nonzero leaf $v$ of $T_{j}'$ has norm
\begin{align*}
\NORM{T_{v}|_{\binom{\{1,2,\ldots,|I_{\ell}|\}}{k_{\ell}}}} & \leq2C^{k_{\ell}}\binom{|I_{\ell}|}{k_{\ell}}^{1/2}\Lambda_{n^{2},k_{\ell}}(\dns(T_{v}))\\
 & \leq2C^{k_{\ell}}\binom{|I_{\ell}|}{k_{\ell}}^{1/2}\Lambda_{n^{2},k_{\ell}}(3^{-j}),
\end{align*}
where the first step applies the inductive hypothesis to the tree
$T_{v}$ of depth $|I_{\ell}|,$ and the second step is legitimate
by the monotonicity of $\Lambda_{n^{2},k_{\ell}}$ (Lemma~\ref{lem:Lambda-cont-monotone-concave}).
Now Lemma~\ref{lem:collect-like-leaves} gives, for each $i=1,2,\ldots,\binom{n}{k_{\ell}},$
a real number $0\leq p_{i}\leq1$ and a decision tree $U_{j,i}\in\Tcal^{*}(n,d',p_{i},0)$
such that
\begin{align}
 & \dns(T_{j}')=\sum_{i=1}^{\binom{n}{k_{\ell}}}p_{i},\label{eq:T-j-prime-probabilities}\\
 & \norm{T_{j}'|_{\Scal'}}\leq2C^{k_{\ell}}\binom{|I_{\ell}|}{k_{\ell}}^{1/2}\Lambda_{n^{2},k_{\ell}}(3^{-j})\sum_{i=1}^{\binom{n}{k_{\ell}}}\norm{U_{j,i}|_{\Scal'}}.\label{eq:T-j-prime-collect}
\end{align}
Applying the inductive hypothesis to each $U_{j,i}|_{\Scal'}$ gives
\begin{align}
\sum_{i=1}^{\binom{n}{k_{\ell}}}\norm{U_{j,i}|_{\Scal'}} & \leq2C^{k-k_{\ell}}\,12^{\ell-2}\sqrt{\binom{|I_{1}|}{k_{1}}\cdots\binom{|I_{\ell-1}|}{k_{\ell-1}}}\,\sum_{i=1}^{\binom{n}{k_{\ell}}}\Lambda_{n^{2},k-k_{\ell}}(p_{i}).\label{eq:U-j-i-apply-IH}
\end{align}
The final summation can be bounded via
\begin{align}
\sum_{i=1}^{\binom{n}{k_{\ell}}}\Lambda_{n^{2},k-k_{\ell}}(p_{i}) & \leq\binom{n}{k_{\ell}}\cdot\Lambda_{n^{2},k-k_{\ell}}\left(\binom{n}{k_{\ell}}^{-1}\sum_{i=1}^{\binom{n}{k_{\ell}}}p_{i}\right)\nonumber \\
 & =n^{2}\cdot\frac{1}{n^{2}}\binom{n}{k_{\ell}}\cdot\Lambda_{n^{2},k-k_{\ell}}\left(\binom{n}{k_{\ell}}^{-1}\dns(T_{j}')\right)\nonumber \\
 & \leq n^{2}\Lambda_{n^{2},k-k_{\ell}}\left(\frac{\dns(T_{j}')}{n^{2}}\right),\label{eq:sum-of-Lambdas}
\end{align}
where the first step is valid by Lemma~\ref{lem:Lambda-cont-monotone-concave}~(iii);
the second step is a substitution from~(\ref{eq:T-j-prime-probabilities});
and the third step uses $k_{\ell}\leq2$ along with Corollary~\ref{cor:Lambda-divided-p-monotone}.
Now
\begin{align*}
\!\!\!\!\!\!\!\!\!\!\!\!\!\norm{T_{j}'|_{\Scal'}} & \leq4C^{k}\,12^{\ell-2}\sqrt{\binom{|I_{1}|}{k_{1}}\cdots\binom{|I_{\ell}|}{k_{\ell}}}\cdot\Lambda_{n^{2},k_{\ell}}(3^{-j})\cdot n^{2}\Lambda_{n^{2},k-k_{\ell}}\left(\frac{\dns(T_{j}')}{n^{2}}\right)\\
 & \leq8C^{k}\,12^{\ell-2}\sqrt{\binom{|I_{1}|}{k_{1}}\cdots\binom{|I_{\ell}|}{k_{\ell}}}\cdot\frac{\Lambda_{n^{2},k_{\ell}}(\sqrt{3^{-j}})}{\sqrt{3^{j}}}\cdot n^{2}\Lambda_{n^{2},k-k_{\ell}}\left(\frac{\dns(T_{j}')}{n^{2}}\right)\\
 & \leq8C^{k}\,12^{\ell-2}\sqrt{\binom{|I_{1}|}{k_{1}}\cdots\binom{|I_{\ell}|}{k_{\ell}}}\cdot\sqrt{3^{-j}}\Lambda_{n^{2},k}(\sqrt{3^{-j}}\dns(T_{j}')),
\end{align*}
where the first step combines~(\ref{eq:T-j-prime-collect})\textendash (\ref{eq:sum-of-Lambdas});
the second step uses~$k_{\ell}\leq2$ and Corollary~\ref{lem:Lambda-sqrt};
and the third step applies Corollary~\ref{cor:Lambda-collapsing}.
The proof of the claim is complete. We note that the final appeal to
Corollary~\ref{cor:Lambda-collapsing} is the reason why our Fourier weight
bound features $\Lambda_{n^2,k}$ rather than
$\Lambda_{n,k}.$
\end{proof}

\subsection{\label{subsec:Generalization-to-elementary}Generalization to elementary
families}

The result on the Fourier spectrum of decision
trees that we have just established
(Theorem~\ref{thm:norm-canonical}) holds only for
elementary families of special form, described at the
beginning of Section~\ref{subsec:Contiguous-intervals}.
We now generalize Theorem~\ref{thm:norm-canonical} to arbitrary
elementary families.
\begin{theorem}
\label{thm:norm-elementary}Let $T\in\Tcal^{*}(n,d,p,0)$ be given,
for some $0\leq p\leq1$ and integers $n,d\geq1$. Let $k$ be an
integer with $1\leq k\leq d.$ Then every elementary family $\Ecal\subseteq\Pcal_{d,k}$
satisfies
\begin{equation}
\NORM{T|_{\Ecal}}\leq(12C)^{k}\Lambda_{n^{2},k}(p)\sqrt{|\Ecal|},\label{eq:norm-elementary}
\end{equation}
where $C\geq1$ is the absolute constant from Theorem~\emph{\ref{thm:tal-depth-2}}.
\end{theorem}

\begin{proof}
If $\Ecal=\varnothing,$ then $T|_{\Ecal}\equiv0$ and the claimed
upper bound holds trivially. In the complementary case of nonempty
$\Ecal$, let $\ell$ be the minimum positive integer such that 
\begin{equation}
\Ecal=\binom{I_{1}}{k_{1}}*\binom{I_{2}}{k_{2}}*\cdots*\binom{I_{\ell}}{k_{\ell}}\label{eq:Ecal-minimality}
\end{equation}
for some pairwise disjoint integer intervals $I_{1},I_{2},\ldots,I_{\ell}$
and some $k_{1},k_{2},\ldots,k_{\ell}\in\{0,1,2\}$. Since $\Ecal\ne\varnothing,$
Proposition~\ref{prop:ast-properties}~(i) implies that $\binom{I_{j}}{k_{j}}\ne\varnothing$
for all $j$ and therefore 
\begin{align}
|I_{j}| & \geq k_{j}, &  & j=1,2,\ldots,\ell.\label{eq:Ij-kj}
\end{align}
The reader will recall from the definition of the $*$ operator that
\begin{align}
 & |\Ecal|=\prod_{j=1}^{\ell}\binom{|I_{j}|}{k_{j}},\label{eq:Ecal-product}\\
 & k=\sum_{j=1}^{\ell}k_{j}.
\end{align}
Since we chose a representation~(\ref{eq:Ecal-minimality}) with
the minimum $\ell,$ Proposition~\ref{prop:ast-properties}~(ii)
additionally implies that $\binom{I_{j}}{k_{j}}\ne\{\varnothing\}$
for all $j,$ forcing 
\begin{align}
k_{j} & \in\{1,2\}, &  & j=1,2,\ldots,\ell.\label{eq:kj-nonzero}
\end{align}
The previous two equations yield
\begin{equation}
\ell\leq k.\label{eq:ell-k}
\end{equation}
It follows from~(\ref{eq:Ij-kj}) and~(\ref{eq:kj-nonzero}) that
each $I_{j}$ is a nonempty subset of $\{1,2,\ldots,d\}$. Furthermore,
by renumbering the intervals if necessary, we may assume that $I_{1}<I_{2}<\cdots<I_{\ell}.$
We abbreviate $I=I_{1}\cup I_{2}\cup\cdots\cup I_{\ell}$ and $\overline{I}=\{1,2,\ldots,d\}\setminus I.$

It is obvious that every string $v\in\pomo^{d}$ is uniquely determined
by its substrings $v|_{I}$ and $v|_{\overline{I}}$. Similarly, for
every $i\in I,$ the prefix $v_{1}v_{2}\ldots v_{i-1}$ is uniquely
determined by the substrings $(v_{1}v_{2}\ldots v_{i-1})|_{I}$ and
$v|_{\overline{I}}.$ This means in particular that
\begin{align}
 & T(v)=U_{v|_{\overline{I}}}(v|_{I}), &  & v\in\pomo^{d}\label{eq:T-change-variable-order-whole-v}\\
 & T(v_{1}v_{2}\ldots v_{i-1})=U_{v|_{\overline{I}}}((v_{1}v_{2}\ldots v_{i-1})|_{I}), &  & v\in\pomo^{d},\;\;i\in I,\label{eq:T-change-variable-order}
\end{align}
where $\{U_{w}\colon w\in\pomo^{|\overline{I}|}\}$ is a suitable
collection of decision trees of depth $I$. By definition,
\begin{align}
U_{w} & \in\Tcal^{*}(n,|I|,\dns(U_{w}),0), &  & w\in\pomo^{|\overline{I}|}.\label{eq:U-w-IH}
\end{align}
Moreover, the densities of the $U_{w}$ are related in a natural way
to the density of $T.$ Indeed, considering a uniformly random string
$v\in\pomo^{d}$ in~(\ref{eq:T-change-variable-order-whole-v}) gives
$\Prob[T(v)\ne0]=\Prob[U_{v|_{\overline{I}}}(v|_{I})\ne0],$ which
is equivalent to
\begin{equation}
\dns(T)=\Exp\dns(U_{v|_{\overline{I}}}).\label{eq:dns-U-w}
\end{equation}

In what follows, all expectations are with respect to uniformly random
$v\in\pomo^{d}.$ We have:
\begin{align*}
T|_{\Ecal} & =\Exp\left[\sum_{S\in\Ecal}T(v)\prod_{i\in S}v_{i}x_{T(v_{1}v_{2}\ldots v_{i-1})}\right]\\
 & =\Exp\left[\sum_{S_{1}\in\binom{I_{1}}{k_{1}}}\cdots\sum_{S_{\ell}\in\binom{I_{\ell}}{k_{\ell}}}T(v)\prod_{j=1}^{\ell}\prod_{i\in S_{j}}v_{i}x_{T(v_{1}v_{2}\ldots v_{i-1})}\right]\\
 & =\Exp\left[\sum_{S_{1}\in\binom{I_{1}}{k_{1}}}\cdots\sum_{S_{\ell}\in\binom{I_{\ell}}{k_{\ell}}}U_{v|_{\overline{I}}}(v|_{I})\prod_{j=1}^{\ell}\prod_{i\in S_{j}}v_{i}x_{U_{v|_{\overline{I}}}((v_{1}v_{2}\ldots v_{i-1})|_{I})}\right],
\end{align*}
where the last step uses~(\ref{eq:T-change-variable-order-whole-v})
and (\ref{eq:T-change-variable-order}). It remains to shift the indexing
variable $i$. For this, let $I_{1}'<I_{2}'<\cdots<I_{\ell}'$ denote
the integer intervals that form a partition of $\{1,2,\ldots,|I|\}$
and satisfy $|I_{j}'|=|I_{j}|$ for all $j.$ Now the previous equation
for $T|_{\Ecal}$ can be restated as
\begin{align}
T|_{\Ecal} & =\Exp\left[\sum_{S_{1}\in\binom{I_{1}'}{k_{1}}}\cdots\sum_{S_{\ell}\in\binom{I_{\ell}'}{k_{\ell}}}U_{v|_{\overline{I}}}(v|_{I})\prod_{j=1}^{\ell}\prod_{i\in S_{j}}(v|_{I})_{i}\cdot x_{U_{v|_{\overline{I}}}((v|_{I})_{\leq i-1})}\right]\nonumber \\
 & =\Exp\left[U_{v|_{\overline{I}}}|_{\binom{I_{1}'}{k_{1}}*\cdots*\binom{I_{\ell}'}{k_{\ell}}}\right].\label{eq:T-Exp-Uw}
\end{align}
As a result,
\begin{align*}
\norm{T|_{\Ecal}} & \leq\Exp\NORM{U_{v|_{\overline{I}}}|_{\binom{I_{1}'}{k_{1}}*\cdots*\binom{I_{\ell}'}{k_{\ell}}}}\\
 & \leq\Exp\left[2C^{k}\,12^{\ell-1}\Lambda_{n^{2},k}(\dns(U_{v|_{\overline{I}}}))\prod_{i=1}^{\ell}\binom{|I_{i}'|}{k_{i}}^{1/2}\right]\\
 & =2C^{k}\,12^{\ell-1}\Exp\left[\Lambda_{n^{2},k}(\dns(U_{v|_{\overline{I}}}))\prod_{i=1}^{\ell}\binom{|I_{i}|}{k_{i}}^{1/2}\right]\\
 & =2C^{k}\,12^{\ell-1}\sqrt{|\Ecal|}\,\Exp\left[\Lambda_{n^{2},k}(\dns(U_{v|_{\overline{I}}}))\right]\\
 & \leq2C^{k}\,12^{\ell-1}\sqrt{|\Ecal|}\,\Lambda_{n^{2},k}(\Exp\dns(U_{v|_{\overline{I}}}))\\
 & \leq(12C)^{k}\sqrt{|\Ecal|}\,\Lambda_{n^{2},k}(\dns(T)),
\end{align*}
where the first step applies Proposition~\ref{prop:fourier-norm-convex}
to~(\ref{eq:T-Exp-Uw}); the second step is justified by~(\ref{eq:U-w-IH})
and Theorem~\ref{thm:norm-canonical}; the fourth step is a substitution
from~(\ref{eq:Ecal-product}); the fifth step is legitimate by Lemma~\ref{lem:Lambda-cont-monotone-concave}~(iii);
and the final step uses~(\ref{eq:ell-k}) and~(\ref{eq:dns-U-w}).
Since $T$ has density $p$ by hypothesis, the proof is complete.
\end{proof}

\subsection{\label{subsec:Main-result}Main result on decision trees}

We now obtain our main result on the Fourier spectrum of decision
trees by combining Theorem~\ref{thm:norm-elementary} with an efficient
decomposition of $\Pcal_{d,k}$ into elementary families (Theorem~\ref{thm:pi-n-choose-k}).
\begin{theorem}
\label{thm:main-fourier-DT-p}Let $f\colon\pomon\to\{-1,0,1\}$ be
a function computable by a decision tree of depth $d.$ Define $p=\Prob_{x\in\pomon}[f(x)\ne0].$
Then 
\begin{align*}
\norm{L_{k}f} & \leq\binom{d}{k}^{1/2}(58Cc)^{k}\,\Lambda_{n^{2},k}(p), &  & k=1,2,\ldots,n,
\end{align*}
where $C\geq1$ and $c\geq1$ are the absolute constants from Theorem~\emph{\ref{thm:tal-depth-2}}
and Lemma~\emph{\ref{lem:sqrt-double-binomial-sum}, respectively.}
\end{theorem}

\begin{proof}
Lemma~\ref{lem:T_S-spectrum} ensures that $L_{k}f=0$ for $k>d,$
so that the theorem holds vacuously in that case. We now examine the
complementary possibility, $1\leq k\leq d.$ For some integer $N\geq1,$
Theorem~\ref{thm:pi-n-choose-k} gives a partition $\Pcal_{d,k}=\bigcup_{i=1}^{N}\Ecal_{i}$
where $\Ecal_{1},\Ecal_{2},\ldots,\Ecal_{N}$ are elementary families
with 
\begin{equation}
\sum_{i=1}^{N}|\Ecal_{i}|^{1/2}\leq(2+2\sqrt{2})^{k}c^{k}\left(\frac{d}{k}\right)^{k/2}.\label{eq:elementary-d-k}
\end{equation}
Fix a decision tree $T$ of depth $d$ that computes $f.$ Then Fact~\ref{fact:dns-prob-nonzero}
shows that $T\in\Tcal^{*}(n,d,p,0).$ As a result,
\begin{align*}
\norm{L_{k}f} & =\norm{T|_{\Pcal_{d,k}}}\\
 & =\NORM{\sum_{i=1}^{N}T|_{\Ecal_{i}}}\\
 & \leq\sum_{i=1}^{N}\norm{T|_{\Ecal_{i}}}\\
 & \leq\sum_{i=1}^{N}(12C)^{k}\,\Lambda_{n^{2},k}(p)\sqrt{|\Ecal_{i}|}\\
 & \leq\left(\frac{d}{k}\right)^{k/2}(58Cc)^{k}\,\Lambda_{n^{2},k}(p),
\end{align*}
where the first step is valid by Lemma~\ref{lem:T_S-spectrum}; the
second step uses Proposition~\ref{prop:T_S-additive}; the third
step uses Proposition~\ref{prop:fourier-norm-convex}; the fourth
step applies Theorem~\ref{thm:norm-elementary}; and the final step
substitutes the upper bound from~(\ref{eq:elementary-d-k}). In view
of~(\ref{eq:entropy-bound-binomial}), the proof is complete.
\end{proof}
\noindent Maximizing over $0\leq p\leq1,$ we establish the following
clean bound conjectured by Tal~\cite{Tal19Query}.
\begin{corollary}
\label{cor:main-Fourier-dt}Let $f\colon\pomon\to\{-1,0,1\}$ be a
function computable by a decision tree of depth $d.$ Then
\begin{align*}
\norm{L_{k}f} & \leq C^{k}\sqrt{\binom{d}{k}(1+\ln n)^{k-1}}, &  & k=1,2,\ldots,n,
\end{align*}
where $C\geq1$ is an absolute constant.
\end{corollary}

\begin{proof}
Recall from Lemma~\ref{lem:Lambda-cont-monotone-concave}~(ii) that
$\Lambda_{n^{2},k}(p)\leq
\Lambda_{n^2,k}(1)=
\sqrt{(\ln en^{2})^{k-1}}$ for all $0\leq p\leq1.$
Now the claimed bound is immediate from Theorem~\ref{thm:main-fourier-DT-p}
after a change of constant~$C.$
\end{proof}
\noindent Corollary~\ref{cor:main-Fourier-dt} settles Theorem~\ref{thm:MAIN-Fourier-bound}
from the introduction. By convexity (Proposition~\ref{prop:fourier-norm-convex}),
Corollary~\ref{cor:main-Fourier-dt} holds more generally for any
real function $f\colon\pomon\to[-1,1]$ computable by a decision tree
of depth $d.$
We record the following generalization for functions with range $[-1,1].$

\begin{corollary}
Let $f\colon\pomon\to[-1,1]$ be a
function computable by a decision tree of depth $d.$ Then
\begin{align*}
\norm{L_{k}f} & \leq C^{k}\sqrt{\binom{d}{k}(1+\ln n)^{k-1}}, &  & k=1,2,\ldots,n,
\end{align*}
where $C\geq1$ is an absolute constant.
\end{corollary}

\begin{proof}
The proof is a reprise of Lemma~\ref{lem:collect-like-leaves}.
Any real number in $[-1,1]$ is a convex combination of $-1$ and $1.$ With this
in mind, the idea is to express a decision tree for $f$
as a convex combination of decision trees with leaf
labels in $\pomo,$ then bound the $k$-Fourier weight
for each of them via
Corollary~\ref{cor:main-Fourier-dt}, and finally infer
a bound on the $k$-Fourier weight of $f$ by convexity.

Formally, let $T$ be a depth-$d$ decision tree that computes $f.$ Let
$\mathbf{T}$ be a random depth-$d$ decision tree with
leaf labels in $\pomo$ such that 
\begin{align*}
&\mathbf{T}(v)=T(v),&&v\in\pomo^{\leq d-1},\\
&\Exp \mathbf{T}(v)=T(v), &&v\in\pomo^d.
\end{align*}
By definition, $\mathbf T$ computes a function 
$f_{\mathbf T}\colon\pomon\to\pomo,$ where $\Exp f_{\mathbf T}=f.$ 
Now for each $k=1,2,\ldots,n,$
\begin{align*}
\norm{L_{k}f} &= \norm{L_{k}\left(\Exp f_{\mathbf T}\right)}\\
  &= \norm{\Exp L_{k} f_{\mathbf T}}\\
  &\leq \Exp \norm{L_{k} f_{\mathbf T}}\\
  & \leq \Exp\; C^{k}\sqrt{\binom{d}{k}(1+\ln n)^{k-1}} \\
  & = C^{k}\sqrt{\binom{d}{k}(1+\ln n)^{k-1}},
\end{align*}
where the second step uses the linearity of $L_k,$ the third step 
applies Proposition~\ref{prop:fourier-norm-convex}, and the fourth 
step is valid for some absolute constant $C\geq1$ by
Corollary~\ref{cor:main-Fourier-dt}.
\end{proof}

\section{Quantum versus classical query complexity}
\label{sec:quantum}

Using our newly derived bound for the Fourier spectrum of decision
trees, we will now prove the main result of this paper on quantum
versus randomized query complexity.

\subsection{Quantum and randomized query models}

For a nonempty finite set $X,$ a \emph{partial Boolean function on
$X$} is a mapping $X\to\{0,1,*\},$ where the output value $*$ is
reserved for illegal inputs. Recall that a \emph{randomized query
algorithm of cost $d$} is a probability distribution on decision
trees of depth at most $d.$ For a (possibly partial) Boolean function
$f$ on the Boolean hypercube, we say that a randomized query algorithm
\emph{computes $f$ with error $\epsilon$} if, for every input $x\in f^{-1}(0)\cup f^{-1}(1),$
the algorithm outputs $f(x)$ with probability at least $1-\epsilon.$
Observe that in this formalism, the algorithm is allowed to exhibit
arbitrary behavior on the illegal inputs, namely, those in $f^{-1}(*).$
The \emph{randomized query complexity $R_{\epsilon}(f)$} is the minimum
cost of a randomized query algorithm that computes $f$ with error
$\epsilon.$ The canonical setting of the error parameter is $\epsilon=1/3.$
This choice is largely arbitrary because the error of a query algorithm
can be reduced in an efficient manner by running the algorithm several
times independently and outputting the majority answer. Quantitatively,
the following relation follows from the Chernoff bound:
\begin{align}
R_{\epsilon}(f) & \leq O\left(\frac{1}{\gamma^{2}}\log\frac{1}{\epsilon}\right)\cdot R_{\frac{1}{2}-\gamma}(f)\label{eq:error-reduction}
\end{align}
for all $\epsilon,\gamma\leq1/2.$

These classical definitions carry over in the obvious way to the quantum
model. Here, the cost is the worst-case number of quantum queries
on any input, and a quantum algorithm is said to \emph{compute $f$
with error $\epsilon$} if, for every input $x\in f^{-1}(0)\cup f^{-1}(1),$
the algorithm outputs $f(x)$ with probability at least $1-\epsilon.$
The \emph{quantum query complexity} $Q_{\epsilon}(f)$ is the minimum
cost of a quantum query algorithm that computes $f$ with error $\epsilon.$
For an excellent introduction to classical and quantum query complexity,
we refer the reader to~\cite{buhrman-dewolf02DT-survey} and~\cite{dewolf-thesis},
respectively.

\subsection{\label{subsec:rorrelation}The rorrelation problem}

We now formally state the problem of interest to us, Tal's \emph{rorrelation}~\cite{Tal19Query},
which was briefly reviewed in the introduction. Let $n$ and $k$
be positive integers. For an orthogonal matrix $U\in\Re^{n\times n},$
consider the multilinear polynomial $\phi_{n,k,U}\colon(\pomon)^{k}\to\Re$
given by
\begin{equation}
\phi_{n,k,U}(x_{1},x_{2},\ldots,x_{k})=\frac{1}{n}\mathbf{1}^{\intercal}D_{x_{1}}UD_{x_{2}}UD_{x_{3}}U\cdots UD_{x_{k}}\mathbf{1},\label{eq:rorrelation-restated}
\end{equation}
where $\1$ denotes the all-ones vector and $D_{x_{i}}$ denotes the
diagonal matrix with vector $x_{i}$ on the diagonal. In what follows,
we treat the sets $(\{-1,1\}^{n})^{k}$ and $\pomo^{n\times k}$ interchangeably,
thereby interpreting the input to $\phi_{n,k,U}$ as an $n\times k$
sign matrix. Let $\|\cdot\|_{2}$ denote the Euclidean norm. Then
for all $x_{1},x_{2},\ldots,x_{k}\in\pomon,$ we have
\begin{align}
|\phi_{n,k,U}(x_{1},x_{2},\ldots,x_{k})| & =\frac{1}{n}\langle\1,D_{x_{1}}UD_{x_{2}}UD_{x_{3}}U\cdots UD_{x_{k}}\1\rangle\nonumber \\
 & \leq\frac{1}{n}\|\1\|_{2}\;\|D_{x_{1}}UD_{x_{2}}UD_{x_{3}}U\cdots UD_{x_{k}}\1\|_{2}\nonumber \\
 & =\frac{1}{n}\|\1\|_{2}\;\|\1\|_{2}\nonumber \\
 & =1,\label{eq:rorrelation-bounded}
\end{align}
where the second step applies the Cauchy\textendash Schwarz inequality,
and the third step is valid because each of the matrices involved
preserves the Euclidean norm. In particular, the multivariate polynomial
$\phi_{n,k,U}$ ranges in $[-1,1]$ for all inputs. Generalizing the
forrelation problem of Aaronson and Ambainis~\cite{AA15forrelation},
Tal~\cite{Tal19Query} considered the partial Boolean function $f_{n,k,U}\colon\pomo^{n\times k}\to\{0,1,*\}$
given by
\[
f_{n,k,U}(x)=\begin{cases}
1 & \text{if }\phi_{n,k,U}(x)\geq2^{-k},\\
0 & \text{if }|\phi_{n,k,U}(x)|\leq2^{-k-1},\\
* & \text{otherwise.}
\end{cases}
\]
Aaronson and Ambainis~\cite{AA15forrelation} showed that there is
a quantum algorithm with $\lceil k/2\rceil$ queries whose acceptance
probability on input $x\in\pomo^{n\times k}$ is $(\phi_{n,k,H}(x)+1)/2$,
where $H$ is the Hadamard transform matrix. Their analysis generalizes
to any orthogonal matrix in place of $H,$ to the following effect.
\begin{fact}[{Tal~\cite[Claim~3.1]{Tal19Query}}]
\label{fact:rorrelation-algorithm} Let $n$ and $k$ be positive
integers, where $n$ is a power of $2.$ Let $U$ be an arbitrary
orthogonal matrix. Then there is a quantum query algorithm with $\lceil k/2\rceil$
queries whose acceptance probability on input $x\in\pomo^{n\times k}$
equals $(\phi_{n,k,U}(x)+1)/2.$
\end{fact}

\begin{corollary}
\label{cor:rorrelation-algorithm}Let $n$ and $k$ be positive integers,
where $n$ is a power of $2.$ Let $U$ be an arbitrary orthogonal
matrix. Then
\begin{equation}
Q_{\frac{1}{2}-\frac{1}{2^{k+4}}}(f_{n,k,U})\leq\left\lceil \frac{k}{2}\right\rceil .\label{eq:q-weak-alg}
\end{equation}
In particular,
\begin{equation}
Q_{1/3}(f_{n,k,U})\leq O(k4^{k}).\label{eq:q-upper-bounded-error}
\end{equation}
\end{corollary}

\begin{proof}
On input $x,$ the query algorithm for~(\ref{eq:q-weak-alg}) is
as follows: with probability $p,$ run the algorithm of Fact~\ref{fact:rorrelation-algorithm}
and output the resulting answer; with complementary probability $1-p,$
output ``no'' regardless of $x$. By design, the proposed solution
has query cost at most $\lceil k/2\rceil$ and accepts $x$ with probability
exactly
\[
p\cdot\frac{\phi_{n,k,U}(x)+1}{2}.
\]
We want this quantity to be at most $\frac{1}{2}-2^{-k-4}$ if $\phi_{n,k,U}(x)\leq2^{-k-1}$,
and at least $\frac{1}{2}+2^{-k-4}$ if $\phi_{n,k,U}(x)\geq2^{-k}.$
These requirements are both met for $p=(1+\frac{3}{2^{k+2}})^{-1}.$
In summary, $f_{n,k,U}$ has a query algorithm with error at most
$\frac{1}{2}-2^{-k-4}$ and query cost $\lceil k/2\rceil.$ To reduce
the error to $1/3,$ run this algorithm independently $\Theta(4^{k})$
times and output the majority answer; cf.~(\ref{eq:error-reduction}).
\end{proof}
Corollary~\ref{cor:rorrelation-algorithm} shows that the rorrelation
problem has small quantum query complexity. By contrast, we will show
that its randomized complexity is essentially the maximum possible.
Specifically, we will prove an optimal, near-linear lower bound on
the randomized query complexity of rorrelation by combining Tal's
work~\cite{Tal19Query} with our near-optimal bounds for the Fourier
spectrum of decision trees.

In what follows, we let $\mathcal{U}_{n,k}$ denote the uniform probability
distribution on $\pomo^{n\times k}.$ Applying Parseval's identity
to the multilinear polynomial $\phi_{n,k,U}$ gives:
\begin{fact}[{Tal~\cite[Claim~4.4]{Tal19Query}}]
\label{fact:phi-Uniform}$\Exp_{x\sim\Ucal_{n,k}}[\phi_{n,k,U}(x)^{2}]=1/n.$
\end{fact}

\noindent The other result from~\cite{Tal19Query} that we will need
is as follows.
\begin{fact}[{Tal~\cite[Lemmas 5.6, 5.7, and Claim~4.1]{Tal19Query}}]
\label{fact:tal-gaussian} Let $n$ be a positive integer. 
Let $U\in\Re^{n\times n}$ be a uniformly random orthogonal matrix.
Then with probability $1-o(1)$ over the choice of $U,$
there exists 
for every positive integer $k$
a probability distribution
$\Dcal_{n,k,U}$ on $\pomo^{n\times k}$ such that:
\begin{align}
 & \Exp_{x\sim\Dcal_{n,k,U}}\phi_{n,k,U}(x)\geq\left(\frac{2}{\pi}\right)^{k-1},\label{eq:Dcal-yes-instances}\\
 & \Exp_{x\sim\Dcal_{n,k,U}}\prod_{(i,j)\in S}x_{i,j}=0, &  & |S|=1,2,\ldots,k-1,\label{eq:Dcal-Fourier-less-than-k}\\
 & \left|\Exp_{x\sim\Dcal_{n,k,U}}\prod_{(i,j)\in S}x_{i,j}\right|\leq\left(\frac{c|S|\log n}{n}\right)^{\frac{|S|}{2}\cdot\frac{k-1}{k}}, &  & |S|=k,k+1,\ldots,nk,\label{eq:Dcal-Fourier-k-or-higher}
\end{align}
where $c\geq1$ is an absolute constant independent of $n,k,U.$
\end{fact}

\subsection{The quantum-classical separation}

In this section, we derive our lower bound on the randomized query
complexity of the rorrelation problem by combining Tal's Facts~\ref{fact:phi-Uniform}
and~\ref{fact:tal-gaussian} with our main result on decision trees~(Corollary~\ref{cor:main-Fourier-dt}).
The technical centerpiece of this derivation is the following ``indistinguishability''
lemma, which is a polynomial improvement on the analogous calculation
by Tal~\cite[Theorem~5.8]{Tal19Query} that used weaker Fourier bounds
for decision trees.
\begin{lemma}
\label{lem:g-not-a-distinguisher}Let $n$ be a positive integer.
Let $U\in\Re^{n\times n}$ be a uniformly random orthogonal matrix.
Then with probability $1-o(1)$ over the choice of $U,$ the following holds for every
integer $k\geq1$ and every function $g\colon\pomo^{n\times k}\to\{0,1\}$$:$
\begin{equation}
\left|\Exp_{\Ucal_{n,k}}g-\Exp_{\Dcal_{n,k,U}}g\right|\leq\left(cd\cdot\frac{\log^{2-\frac{1}{k}}(n+k)}{n^{1-\frac{1}{k}}}\right)^{k/2},\label{eq:g-not-a-distinguisher}
\end{equation}
where $\Dcal_{n,k,U}$ is as defined in Fact~\emph{\ref{fact:tal-gaussian};}
$\,\,d$ is the minimum depth of a decision tree that computes $g;$
and $c\geq1$ is an absolute constant independent of $n,k,U,g.$
\end{lemma}

\begin{proof}
Fact~\ref{fact:tal-gaussian} guarantees that with probability $1-o(1)$
over the choice of $U,$ there exists for every integer $k\geq1$
a probability distribution $\Dcal_{n,k,U}$ on $\pomo^{n\times k}$
that obeys~(\ref{eq:Dcal-yes-instances})\textendash (\ref{eq:Dcal-Fourier-k-or-higher}).
Conditioned on this event, we will prove~(\ref{eq:g-not-a-distinguisher}).
To start with, fix $g$ and write out the Fourier expansion
\begin{align*}
g(x) & =\sum_{S\subseteq\{1,2,\ldots,n\}\times\{1,2,\ldots,k\}}\hat{g}(S)\prod_{(i,j)\in S}x_{i,j}\\
 & =\sum_{\ell=0}^{nk}\sum_{|S|=\ell}\hat{g}(S)\prod_{(i,j)\in S}x_{i,j}.
\end{align*}
Then
\begin{align*}
\left|\Exp_{\Ucal_{n,k}}g-\Exp_{\Dcal_{n,k,U}}g\right| & \leq\sum_{\ell=0}^{nk}\sum_{|S|=\ell}|\hat{g}(S)|\left|\Exp_{\Ucal_{n,k}}\prod_{(i,j)\in S}x_{i,j}-\Exp_{\Dcal_{n,k,U}}\prod_{(i,j)\in S}x_{i,j}\right|\\
 & \leq\sum_{\ell=1}^{nk}\sum_{|S|=\ell}|\hat{g}(S)|\left|\Exp_{\Ucal_{n,k}}\prod_{(i,j)\in S}x_{i,j}-\Exp_{\Dcal_{n,k,U}}\prod_{(i,j)\in S}x_{i,j}\right|\\
 & \leq\sum_{\ell=k}^{nk}\sum_{|S|=\ell}|\hat{g}(S)|\left|\Exp_{\Dcal_{n,k,U}}\prod_{(i,j)\in S}x_{i,j}\right|,
\end{align*}
where the first step is an application of the triangle inequality; the second step
is justified by $\Exp_{\Ucal_{n,k}}1=\Exp_{\Dcal_{n,k,U}}1=1;$ and
the third step is valid due to~(\ref{eq:Dcal-Fourier-less-than-k})
and the identity $\Exp_{\Ucal_{n,k}}\prod_{(i,j)\in S}x_{i,j}=0$
for nonempty $S$. Let $d$ be the minimum depth of a decision tree
that computes $g.$ Applying~(\ref{eq:Dcal-Fourier-k-or-higher})
then Corollary~\ref{cor:main-Fourier-dt}, we conclude that
\[
\left|\Exp_{\Ucal_{n,k}}g-\Exp_{\Dcal_{n,k,U}}g\right|\leq\sum_{\ell=k}^{nk}c_{1}^{\ell}\sqrt{\binom{d}{\ell}(1+\ln nk)^{\ell-1}}\left(\frac{c_{2}\ell\log n}{n}\right)^{\frac{\ell}{2}\cdot\frac{k-1}{k}},
\]
where $c_{1}\geq1$ and $c_{2}\geq1$ are the absolute constants in
Corollary~\ref{cor:main-Fourier-dt} and Fact~\ref{fact:tal-gaussian}.
In view of~(\ref{eq:entropy-bound-binomial}), this gives
\begin{align*}
\left|\Exp_{\Ucal_{n,k}}g-\Exp_{\Dcal_{n,k,U}}g\right| & \leq\sum_{\ell=k}^{\infty}\left(c_{1}^{2}\cdot\frac{\e d}{\ell}\cdot(1+\ln nk)^{\frac{\ell-1}{\ell}}\cdot\left(\frac{c_{2}\ell\log n}{n}\right)^{\frac{k-1}{k}}\right)^{\frac{\ell}{2}}\\
 & \leq\sum_{\ell=k}^{\infty}\left(c_{1}^{2}\cdot\e d\cdot(1+\ln nk)\cdot\left(\frac{c_{2}\log n}{n}\right)^{\frac{k-1}{k}}\right)^{\frac{\ell}{2}}\\
 & \leq\sum_{\ell=k}^{\infty}\left(\frac{cd}{4}\cdot\frac{\log^{2-\frac{1}{k}}(n+k)}{n^{1-\frac{1}{k}}}\right)^{\frac{\ell}{2}},
\end{align*}
where $c\geq1$ in the last step is a sufficiently large absolute
constant. This settles~(\ref{eq:g-not-a-distinguisher}) in the case
when $cd\log^{(2k-1)/k}(n+k)\leq n^{(k-1)/k}.$ In the complementary
case, (\ref{eq:g-not-a-distinguisher}) follows from the trivial bound
$|\Exp_{\Ucal_{n,k}}g-\Exp_{\Dcal_{n,k,U}}g|\leq1.$
\end{proof}
We have reached the main result of this section, an essentially tight
lower bound on the randomized query complexity of the $k$-fold rorrelation
problem.

\begin{theorem}
\label{thm:randomized-lower-bound} Let $n$ be a positive integer.
Let $U\in\Re^{n\times n}$ be a
uniformly random orthogonal matrix. Then with probability $1-o(1)$ 
over the choice of $U,$ the following holds for all positive integers $k\leq\frac{1}{3}\log n-1$$:$
\begin{equation}
R_{1/2^{k+1}}(f_{n,k,U})=\Omega\left(\frac{n^{1-\frac{1}{k}}}{(\log n)^{2-\frac{1}{k}}}\right),\label{eq:R-lower-bound}
\end{equation}
and in particular
\begin{align}
R_{\frac{1}{2}-\gamma}(f_{n,k,U}) & =\Omega\left(\frac{\gamma^{2}}{k}\cdot\frac{n^{1-\frac{1}{k}}}{(\log n)^{2-\frac{1}{k}}}\right), &  & 0\leq\gamma\leq\frac{1}{2}.\label{eq:R-lower-bound-general}
\end{align}
\end{theorem}

\begin{proof}
We will prove the lower bounds for every $U$ that satisfies (\ref{eq:Dcal-yes-instances})
and~(\ref{eq:g-not-a-distinguisher}) for all $k\geq1,$ which happens with probability
$1-o(1)$ by Fact~\ref{fact:tal-gaussian} and Lemma~\ref{lem:g-not-a-distinguisher}.
To begin with,
\begin{align}
\Prob_{\Ucal_{n,k}}[f_{n,k,U}(x)\ne0] & =\Prob_{\Ucal_{n,k}}[|\phi_{n,k,U}(x)|>2^{-k-1}]\nonumber \\
 & \leq4^{k+1}\Exp_{\Ucal_{n,k}}[\phi_{n,k,U}(x)^{2}]\nonumber \\
 & \leq\frac{4^{k+1}}{n}\nonumber \\
 & \leq\frac{1}{2^{k+1}},\label{eq:U-supported-on-no}
\end{align}
where the last three steps use Markov's inequality, Fact~\ref{fact:phi-Uniform},
and $k\leq\frac{1}{3}\log n-1,$ respectively. Also,
\begin{align*}
\left(\frac{2}{\pi}\right)^{k-1} & \leq\Exp_{\Dcal_{n,k,U}}\phi_{n,k,U}(x)\\
 & \leq2^{-k}\Prob_{\Dcal_{n,k,U}}[\phi_{n,k,U}(x)<2^{-k}]+\Prob_{\Dcal_{n,k,U}}[\phi_{n,k,U}(x)\geq2^{-k}]\\
 & =2^{-k}(1-\Prob_{\Dcal_{n,k,U}}[f_{n,k,U}(x)=1])+\Prob_{\Dcal_{n,k,U}}[f_{n,k,U}(x)=1]\\
 & =2^{-k}+(1-2^{-k})\Prob_{\Dcal_{n,k,U}}[f_{n,k,U}(x)=1],
\end{align*}
where the first and second steps are justified by~(\ref{eq:Dcal-yes-instances})
and~(\ref{eq:rorrelation-bounded}), respectively. The last equation
shows that
\begin{align}
\Prob_{\Dcal_{n,k,U}}[f_{n,k,U}(x)=1] & \geq\left(\frac{2}{\pi}\right)^{k-1}-2^{-k}\nonumber \\
 & \geq2^{-k}.\label{eq:D-supported-on-yes}
\end{align}

Now fix arbitrary parameters $d\geq1$ and $0\leq\epsilon\leq1/2,$
and consider a randomized query algorithm of cost $d$ that computes
$f_{n,k,U}$ with error at most $\epsilon.$ Then the algorithm's
acceptance probability on given input $x$ is $\Exp_{r}g_{r}(x),$
where $r$ denotes a random string and each $g_{r}\colon\pomo^{n\times k}\to\zoo$
is computable by a decision tree of depth at most $d.$ Since the
error is at most $\epsilon,$ we have 
\begin{equation}
\Prob_{r}[f_{n,k,U}(x)=0,\,g_{r}(x)=1]+\Prob_{r}[f_{n,k,U}(x)=1,\,g_{r}(x)=0]\leq\epsilon\label{eq:f-g-error}
\end{equation}
for every $x\in\pomo^{n\times k}.$ We thus obtain the two inequalities
\begin{align}
 & \Exp_{r}\Prob_{\Ucal_{n,k}}[f_{n,k,U}(x)=0,\,g_{r}(x)=1]\leq\epsilon,\label{eq:error-U}\\
 & \Exp_{r}\Prob_{\Dcal_{n,k,U}}[f_{n,k,U}(x)=1,\,g_{r}(x)=0]\leq\epsilon,\label{eq:error-D}
\end{align}
by passing to expectations in~(\ref{eq:f-g-error}) with respect
to $x\sim\Ucal_{n,k}$ and $x\sim\Dcal_{n,k,U},$ respectively. On
the other hand, (\ref{eq:g-not-a-distinguisher}) and $k=O(\log n)$
imply
\begin{equation}
\Exp_{r}\left|\Exp_{\Dcal_{n,k,U}}g_{r}-\Exp_{\Ucal_{n,k}}g_{r}\right|\leq\left(c'd\cdot\frac{(\log n)^{2-\frac{1}{k}}}{n^{1-\frac{1}{k}}}\right)^{\frac{k}{2}}\label{eq:g_r-not-a-distinguisher}
\end{equation}
for some absolute constant $c'\geq1.$

We now have all the ingredients to complete the proof. For each $r,$
we have
\begin{align}
\Exp_{\Dcal_{n,k,U}}g_{r} & =\Prob_{\Dcal_{n,k,U}}[g_{r}(x)=1]\nonumber \\
 & \geq\Prob_{\Dcal_{n,k,U}}[f_{n,k,U}(x)=1]-\Prob_{\Dcal_{n,k,U}}[f_{n,k,U}(x)=1,\;g_{r}(x)=0]\nonumber \\
 & \geq2^{-k}-\Prob_{\Dcal_{n,k,U}}[f_{n,k,U}(x)=1,\;g_{r}(x)=0],\label{eq:fixed-r-random-D}
\end{align}
where the last step uses~(\ref{eq:D-supported-on-yes}). Similarly,
\begin{align}
\Exp_{\Ucal_{n,k}}g_{r} & =\Prob_{\Ucal_{n,k}}[g_{r}(x)=1]\nonumber \\
 & \leq\Prob_{\Ucal_{n,k}}[f_{n,k,U}(x)\ne0]+\Prob_{\Ucal_{n,k}}[f_{n,k,U}(x)=0,\;g_{r}(x)=1]\nonumber \\
 & \leq2^{-k-1}+\Prob_{\Ucal_{n,k}}[f_{n,k,U}(x)=0,\;g_{r}(x)=1],\label{eq:fixed-r-random-U}
\end{align}
where the last step uses~(\ref{eq:U-supported-on-no}). Passing to
expectations in (\ref{eq:fixed-r-random-D}) and~(\ref{eq:fixed-r-random-U})
with respect to $r$ gives
\begin{multline*}
\Exp_{r}\left[\Exp_{\Dcal_{n,k,U}}g_{r}-\Exp_{\Ucal_{n,k}}g_{r}\right]\geq2^{-k-1}-\Exp_{r}\Prob_{\Dcal_{n,k,U}}[f_{n,k,U}(x)=1,\;g_{r}(x)=0]\\
-\Exp_{r}\Prob_{\Ucal_{n,k}}[f_{n,k,U}(x)=0,\;g_{r}(x)=1],
\end{multline*}
which in view of~(\ref{eq:error-U}) and~(\ref{eq:error-D}) simplifies
to
\begin{align*}
\Exp_{r}\left[\Exp_{\Dcal_{n,k,U}}g_{r}-\Exp_{\Ucal_{n,k}}g_{r}\right] & \geq2^{-k-1}-2\epsilon.
\end{align*}
Comparing this lower bound with~(\ref{eq:g_r-not-a-distinguisher}),
we arrive at
\[
\left(c'd\cdot\frac{(\log n)^{2-\frac{1}{k}}}{n^{1-\frac{1}{k}}}\right)^{\frac{k}{2}}\geq2^{-k-1}-2\epsilon.
\]
Taking $\epsilon=2^{-k-3}$ and solving for $d,$ we find that
\[
R_{2^{-k-3}}(f_{n,k,U})=\Omega\left(\frac{n^{1-\frac{1}{k}}}{(\log n)^{2-\frac{1}{k}}}\right).
\]
By the error reduction formula~(\ref{eq:error-reduction}), this
settles~(\ref{eq:R-lower-bound}) and~(\ref{eq:R-lower-bound-general}).
\end{proof}
Theorem~\ref{thm:randomized-lower-bound} settles Theorem~\ref{thm:MAIN-randomized-lower-bound}
from the introduction. Corollary~\ref{cor:bounded-error-t-separation}
now follows from (\ref{eq:q-weak-alg}) and Theorem~\ref{thm:MAIN-randomized-lower-bound}
by taking $k=2t$ and $\gamma=1/6$. Similarly, Corollary~\ref{cor:MAIN-bounded-separation}
follows from (\ref{eq:q-upper-bounded-error}) and Theorem~\ref{thm:MAIN-randomized-lower-bound}
by taking $k=\lceil1/\epsilon\rceil+1$ and $\gamma=1/6$. Finally,
Corollary~\ref{cor:MAIN-near-linear-separation} follows from (\ref{eq:q-upper-bounded-error})
and Theorem~\ref{thm:MAIN-randomized-lower-bound} by setting $\gamma=1/6$
and taking $k=k(n)$ to be a sufficiently slow-growing function.

\section*{Acknowledgments}

The authors are thankful to Nikhil Bansal, Dmitry Gavinsky, Makrand
Sinha, Avishay Tal, and Ronald de Wolf for valuable comments on an
earlier version of this paper. We are additionally thankful to Nikhil
and Makrand for reminding us that quantum-classical query separations
automatically imply analogous separations in communication complexity.

\bibliographystyle{siamplain}
\bibliography{refs}

\begin{thebibliography}{10}

\bibitem{AA15forrelation}
{\sc S.~Aaronson and A.~Ambainis}, {\em Forrelation: {A} problem that optimally
  separates quantum from classical computing}, {SIAM} J. Comput., 47 (2018),
  pp.~982--1038, \url{https://doi.org/10.1137/15M1050902}.

\bibitem{ABK16cheat}
{\sc S.~Aaronson, S.~Ben{-}David, and R.~Kothari}, {\em Separations in query
  complexity using cheat sheets}, in \textit{Proceedings of the Forty-Eighth
  Annual ACM Symposium on Theory of Computing} \textup{(STOC)}, 2016,
  pp.~863--876, \url{https://doi.org/10.1145/2897518.2897644}.

\bibitem{ABKT20}
{\sc S.~Aaronson, S.~Ben{-}David, R.~Kothari, and A.~Tal}, {\em Quantum
  implications of {Huang's} sensitivity theorem}.
\newblock Available at \url{https://arxiv.org/abs/2004.13231}, 2020.

\bibitem{bansal-sinha20forrelation}
{\sc N.~Bansal and M.~Sinha}, {\em {$k$}-forrelation optimally separates
  quantum and classical query complexity}.
\newblock Available at \url{https://arxiv.org/abs/2008.07003}, 2020.

\bibitem{beals-et-al01quantum-by-polynomials}
{\sc R.~Beals, H.~Buhrman, R.~Cleve, M.~Mosca, and R.~de~Wolf}, {\em Quantum
  lower bounds by polynomials}, J. ACM, 48 (2001), pp.~778--797,
  \url{https://doi.org/10.1145/502090.502097}.

\bibitem{bernstein-vazirani93}
{\sc E.~Bernstein and U.~V. Vazirani}, {\em Quantum complexity theory}, {SIAM}
  J. Comput., 26 (1997), pp.~1411--1473,
  \url{https://doi.org/10.1137/S0097539796300921}.

\bibitem{BIJLSV21f2poly}
{\sc J.~Blasiok, P.~Ivanov, Y.~Jin, C.~H. Lee, R.~A. Servedio, and E.~Viola},
  {\em Fourier growth of structured {$\mathbb{F}_2$}-polynomials and
  applications}, in \textit{Proceedings of the Twenty-Fifth International
  Workshop on Randomization and Computation} \textup{(RANDOM)}, vol.~207 of
  LIPIcs, 2021, pp.~53:1--53:20,
  \url{https://doi.org/10.4230/LIPIcs.APPROX/RANDOM.2021.53}.

\bibitem{bravyi2021quantum-to-classical}
{\sc S.~Bravyi, D.~Gosset, D.~Grier, and L.~Schaeffer}, {\em Classical
  algorithms for {F}orrelation}.
\newblock Available at \url{https://arxiv.org/abs/2102.06963}, 2021.
\newblock October 2021.

\bibitem{bcw98quantum}
{\sc H.~Buhrman, R.~Cleve, and A.~Wigderson}, {\em Quantum vs.~classical
  communication and computation}, in \textit{Proceedings of the Thirtieth
  Annual ACM Symposium on Theory of Computing} \textup{(STOC)}, 1998,
  pp.~63--68, \url{https://doi.org/10.1145/276698.276713}.

\bibitem{buhrman-dewolf02DT-survey}
{\sc H.~Buhrman and R.~de~Wolf}, {\em Complexity measures and decision tree
  complexity: {A} survey}, Theor. Comput. Sci., 288 (2002), pp.~21--43,
  \url{https://doi.org/10.1016/S0304-3975(01)00144-X}.

\bibitem{BFNR03quantum-property-testing}
{\sc H.~Buhrman, L.~Fortnow, I.~Newman, and H.~R{\"{o}}hrig}, {\em Quantum
  property testing}, {SIAM} J. Comput., 37 (2008), pp.~1387--1400,
  \url{https://doi.org/10.1137/S0097539704442416}.

\bibitem{cfkmp19lifting-randomized-with-IP-gadget}
{\sc A.~Chattopadhyay, Y.~Filmus, S.~Koroth, O.~Meir, and T.~Pitassi}, {\em
  Query-to-communication lifting for {BPP} using inner product}, in
  \textit{Proceedings of the Forty-Sixth International Colloquium on Automata,
  Languages and Programming} \textup{(ICALP)}, vol.~132 of LIPIcs, 2019,
  pp.~35:1--35:15, \url{https://doi.org/10.4230/LIPIcs.ICALP.2019.35}.

\bibitem{CHHL18walk}
{\sc E.~Chattopadhyay, P.~Hatami, K.~Hosseini, and S.~Lovett}, {\em
  Pseudorandom generators from polarizing random walks}, in \textit{Proceedings
  of the Thirty-Third Annual IEEE Conference on Computational Complexity}
  \textup{(CCC)}, vol.~102, 2018, pp.~1:1--1:21,
  \url{https://doi.org/10.4230/LIPIcs.CCC.2018.1}.

\bibitem{CHRT18ROBP}
{\sc E.~Chattopadhyay, P.~Hatami, O.~Reingold, and A.~Tal}, {\em Improved
  pseudorandomness for unordered branching programs through local
  monotonicity}, in \textit{Proceedings of the Fiftieth Annual ACM Symposium on
  Theory of Computing} \textup{(STOC)}, 2018, pp.~363--375,
  \url{https://doi.org/10.1145/3188745.3188800}.

\bibitem{deutschjozsa1992rapid}
{\sc D.~Deutsch and R.~Jozsa}, {\em Rapid solution of problems by quantum
  computation}, Proc. R. Soc. Lond. A, 439 (1992), pp.~553--558,
  \url{https://doi.org/10.1098/rspa.1992.0167}.

\bibitem{gavinsky20quantum-vs-classical-comm}
{\sc D.~Gavinsky}, {\em Entangled simultaneity versus classical interactivity
  in communication complexity}, {IEEE} Trans. Inf. Theory, 66 (2020),
  pp.~4641--4651, \url{https://doi.org/10.1109/TIT.2020.2976074}.

\bibitem{GTW21parity}
{\sc U.~Girish, A.~Tal, and K.~Wu}, {\em Fourier growth of parity decision
  trees}, in \textit{Proceedings of the Thirty-Sixth Annual IEEE Conference on
  Computational Complexity} \textup{(CCC)}, vol.~200 of LIPIcs, 2021,
  pp.~39:1--39:36, \url{https://doi.org/10.4230/LIPIcs.CCC.2021.39}.

\bibitem{GSW16sensitivity}
{\sc P.~Gopalan, R.~A. Servedio, and A.~Wigderson}, {\em Degree and
  sensitivity: Tails of two distributions}, in \textit{Proceedings of the
  Thirty-First Annual IEEE Conference on Computational Complexity}
  \textup{(CCC)}, vol.~50, 2016, pp.~13:1--13:23,
  \url{https://doi.org/10.4230/LIPIcs.CCC.2016.13}.

\bibitem{grover96search}
{\sc L.~K. Grover}, {\em A fast quantum mechanical algorithm for database
  search}, in \textit{Proceedings of the Twenty-Eighth Annual ACM Symposium on
  Theory of Computing} \textup{(STOC)}, 1996, pp.~212--219,
  \url{https://doi.org/10.1145/237814.237866}.

\bibitem{huang2019induced}
{\sc H.~Huang}, {\em Induced subgraphs of hypercubes and a proof of the
  sensitivity conjecture}, Annals of Mathematics, 190 (2019), pp.~949--955,
  \url{https://doi.org/10.4007/annals.2019.190.3.6}.

\bibitem{jukna11extremal-2nd-edition}
{\sc S.~Jukna}, {\em Extremal Combinatorics with Applications in Computer
  Science}, Springer-Verlag Berlin Heidelberg, 2nd~ed., 2011,
  \url{https://doi.org/10.1007/978-3-642-17364-6}.

\bibitem{LPV22branchingprogram}
{\sc C.~H. Lee, E.~Pyne, and S.~P. Vadhan}, {\em Fourier growth of regular
  branching programs}, in Electronic Colloquium on Computational Complexity
  (ECCC), 2022.
\newblock Report TR22-034.

\bibitem{odonnell14boolean-fuction-analysis}
{\sc R.~O'Donnell}, {\em Analysis of {B}oolean Functions}, Cambridge University
  Press, 2014.

\bibitem{OS07}
{\sc R.~O'Donnell and R.~A. Servedio}, {\em Learning monotone decision trees in
  polynomial time}, {SIAM} J. Comput., 37 (2007), pp.~827--844,
  \url{https://doi.org/10.1137/060669309}.

\bibitem{raz99quantum-classical}
{\sc R.~Raz}, {\em Exponential separation of quantum and classical
  communication complexity}, in \textit{Proceedings of the Thirty-First Annual
  ACM Symposium on Theory of Computing} \textup{(STOC)}, 1999, pp.~358--367,
  \url{https://doi.org/10.1145/301250.301343}.

\bibitem{regev-klartag11quantum-vs-classical}
{\sc O.~Regev and B.~Klartag}, {\em Quantum one-way communication can be
  exponentially stronger than classical communication}, in \textit{Proceedings
  of the Forty-Third Annual ACM Symposium on Theory of Computing}
  \textup{(STOC)}, 2011, pp.~31--40,
  \url{https://doi.org/10.1145/1993636.1993642}.

\bibitem{shor94factoring}
{\sc P.~W. Shor}, {\em Polynomial-time algorithms for prime factorization and
  discrete logarithms on a quantum computer}, SIAM J. Comput., 26 (1997),
  pp.~1484--1509, \url{https://doi.org/10.1137/S0097539795293172}.

\bibitem{simon97power}
{\sc D.~R. Simon}, {\em On the power of quantum computation}, {SIAM} J.
  Comput., 26 (1997), pp.~1474--1483,
  \url{https://doi.org/10.1137/S0097539796298637}.

\bibitem{SVW17}
{\sc T.~Steinke, S.~P. Vadhan, and A.~Wan}, {\em Pseudorandomness and
  {Fourier}-growth bounds for width-3 branching programs}, Theory Comput., 13
  (2017), pp.~1--50, \url{https://doi.org/10.4086/toc.2017.v013a012}.

\bibitem{Tal17ac0}
{\sc A.~Tal}, {\em Tight bounds on the fourier spectrum of {AC0}}, in
  \textit{Proceedings of the Thirty-Second Annual IEEE Conference on
  Computational Complexity} \textup{(CCC)}, vol.~79, 2017, pp.~15:1--15:31,
  \url{https://doi.org/10.4230/LIPIcs.CCC.2017.15}.

\bibitem{Tal19Query}
{\sc A.~Tal}, {\em Towards optimal separations between quantum and randomized
  query complexities}, in \textit{Proceedings of the Sixty-First Annual IEEE
  Symposium on Foundations of Computer Science} \textup{(FOCS)}, 2020,
  pp.~228--239, \url{https://doi.org/10.1109/FOCS46700.2020.00030}.

\bibitem{dewolf-thesis}
{\sc R.~{\noopsort{Wolf}}{de Wolf}}, {\em Quantum Computing and Communication
  Complexity}, PhD thesis, University of Amsterdam, 2001.

\end{thebibliography}

\end{document}